%% file: broadcast-arxiv.tex
\documentclass[]{article}

\usepackage{authblk}
\usepackage{courier}
\usepackage{url}
\usepackage{verbatim} 
\usepackage{graphicx,dsfont,subfigure}
\usepackage{multirow}
\usepackage[noend]{algorithmic}
\usepackage[vlined,ruled]{algorithm2e}
\usepackage{xspace}
\usepackage{amsmath}
\usepackage{amssymb}
\usepackage{amsthm}
\usepackage{Definitions}
\usepackage{empheq}
\usepackage{lipsum}
\usepackage{fullpage}
\usepackage{color}
\usepackage{xspace}
\usepackage{paralist}
\usepackage{enumitem}

\newcommand{\xhdr}[1]{\vspace{1.7mm}\noindent{{\bf #1.}}}

\newcommand\blfootnote[1]{%
  \begingroup
  \renewcommand\thefootnote{}\footnote{#1}%
  \addtocounter{footnote}{-1}%
  \endgroup
}

\begin{document}

\title{Smart broadcasting: Do you want to be seen?}

\author[1]{Mohammad Reza Karimi$^{*}$}
\author[1]{Erfan Tavakoli$^{*}$}
\author[2]{Mehrdad Farajtabar}
\author[2]{\\Le Song}
\author[3]{Manuel Gomez-Rodriguez}

\affil[1]{Sharif University, mkarimi@ce.sharif.edu, erfan.tavakoli71@gmail.com}
\affil[2]{Georgia Institute of Technology, mehrdad@gatech.edu, lsong@cc.gatech.edu}
\affil[3]{Max Planck Institute for Software Systems, manuelgr@mpi-sws.org}

\date{}

\begin{small}
\maketitle
\end{small}

\begin{abstract}
\input{000abstract}
\end{abstract}

\blfootnote{$^{*}$\scriptsize Authors contributed equally. This work was done during the authors'{} internships at Max Planck
for Software Systems.}

\vspace{-1mm}
\section{Introduction}
\label{sec:intro}
\input{010intro}

\vspace{-1mm}
\section{Background on Point Processes}
\label{sec:backgrond}
\input{020background}

\vspace{-1mm}
\section{From Intensities to Visibility}
\label{sec:visibility}
\input{030visibility}

\vspace{-1mm}
\section{On the Concavity of Visibility}
\label{sec:concavity}
\input{040concavity}

\vspace{-1mm}
\section{Convex Visibility Shaping Framework}
\label{sec:formulation}
\input{050formulation}

\vspace{-1mm}
\section{Scalable Algorithm}
\label{sec:algorithm}
\input{060algorithm}

\vspace{-1mm}
\section{Experiments}
\label{sec:experiments}
\input{070experiments}

\vspace{-1mm}
\section{Conclusions}
\label{sec:conclusions}
\input{080conclusions}

\bibliographystyle{abbrv}
\bibliography{refs}

\vspace{-3mm}
\begin{appendix}
\label{sec:appendix}
\input{090appendix}

\end{appendix}

\end{document}

%% file: 000abstract.tex
Many users in online social networks are constantly trying to gain attention from their followers by \emph{broadcasting} posts to them. 
These broadcasters are likely to gain greater attention if their posts can remain visible for a longer period of time among their followers' most recent feeds. Then when to post? 
%
In this paper, we study the problem of smart broad\-cas\-ting using the framework of temporal point processes, where we model users feeds and posts as discrete events occu\-rring in continuous time. Based on such continuous-time model, then choosing a broadcasting strategy for a user becomes a problem of designing the conditional intensity of her posting events. 
We derive a novel formula which links this conditional intensity with the ``visibility'' of the user in her followers' feeds. Furthermore, by exploiting this formula, we develop an efficient convex optimization framework for the ``when-to-post'' problem. Our method can find broadcasting strategies that reach a desired ``visibility'' level with provable guarantees. We experimented with data ga\-the\-red from Twitter, and show that our framework can consistently make broadcasters' post more visible than alternatives.

%% file: 010intro.tex
The popularization of social media and online social networking has empowered political parties, small and large corporations, celebrities, 
as well as ordinary people, with a platform to build, reach and broadcast information to their own audience. 
For example, political leaders use social media to present their character and personalize their message in hopes of tapping
younger voters\footnote{\scriptsize \url{http://www.nytimes.com/2012/10/08/technology/campaigns-use-social-media-to-lure-younger-voters.html}}; 
corporations increasingly rely on social media for a variety of tasks, from brand awareness to marketing and customer 
service~\cite{constantinides2014foundations};
celebrities leverage social media to bring awareness to themselves and strengthen their fans'{} loyalty\footnote{\scriptsize \url{http://www.wsj.com/articles/what-celebrities-can-teach-companies-about-social-media-1444788220}}; 
and, ordinary people post about their lives and express their opinions to gain recognition from a mix of close friends and acquaintances\footnote{\scriptsize \url{http://www.pewinternet.org/topics/social-networking/}}.
However, social media users often \emph{follow} hundreds of \emph{broadcasters}, and they often receive information at a rate far higher than their cognitive abilities to process it~\cite{gomez14icwsm}. This also means that many broadcasters actually share quite a portion of their followers, and they are constantly competing for attention from these followers.

In this context, these followers'{} attention becomes a scarce commodity of great value~\cite{crawford2015world}, 
and broadcasters would like to consume a good share of it so that their posted contents are \emph{noticed} and possibly liked or shared. 
As a consequence, there are myriads of articles and blog entries about the \emph{best times} to broadcast information in social media
and social networking, as well as data analytics tools to find these times\footnote{\scriptsize \url{http://www.huffingtonpost.com/catriona-pollard/the-best-times-to-post-on\_b\_6990376.html}}\footnote{\scriptsize \url{http://blog.klout.com/2015/07/whens-the-best-time-to-post-on-social/}}.
%
%
%
However, the best time to post on social media depends on a variety of factors, often specific to the broadcaster in question, such as their followers'{} daily and weekly behavior patterns, their location or timezone, and 
the number of broadcasters and volume of information competing for their attention in these followers' feeds (be it in the form of a Twitter user'{}s timeline, a Facebook user'{}s wall or an Instagram user'{}s feed).
Therefore, the problem of finding the best times to broadcast messages and elicit attention (be it views, likes or shares), in short, the \emph{when-to-post} problem, requires careful reasoning and smart algorithms, which have been largely inexistent until very recently~\cite{spasojevic2015post}.

In this paper, we develop a novel framework for the when-to-post problem, where we measure the gained attention or \emph{visibility} of a broadcaster as the time that at least one post from her is among the most recent $k$ received stories in her followers'{} feed. A desirable property of this time based vi\-si\-bi\-li\-ty measure is that it is easy to estimate from real data. In order to measure the achieved visibility for a particular deployed broadcasting strategy, one only need to use a separate held-out set of the followers'{} feeds, independently of the broadcasted content. 
This is in contrast to other measures based on,~\eg,~the number of likes or shares caused by a broadcasting strategy. These latter measures are difficult to estimate from real data and often require actual interventions, since they depend on other confounding factors such as the follower'{}s 
reaction to the post content~\cite{chenhao14}, whose effect is difficult to model accurately~\cite{cheng2014can}.

More specifically, we will model users'{} feeds and posts as discrete events occurring in continuous time using the framework of temporal point processes. Our model explicitly characterize the continuous time interval between posts by means of conditional intensity functions~\cite{AalBorGje08}. 
Based on such continuous-time model, then choosing a strategy for a broadcaster becomes a problem of designing
the conditional intensity of her posting events. We derive a novel formula which can link the conditional intensity of an \emph{arbitrary} broadcaster with her visibility in her followers' feeds. Interestingly, 
we can show that the average visibility is concave in the space of (piece-wise) smooth intensity functions. 
Based on this result, we propose a convex optimization framework to address\- a diverse range of visibility sha\-ping tasks given bud\-get constraints. 
Our framework allows us to conduct fine-grained control of a broadcaster's visibility across her followers. For instance, our framework can steer the visibility in such a way that some time intervals are favored over others, \eg, times when the broad\-cas\-ters'{} follo\-wers are on-line.
%
%
In addition to the novel framework, we develop an efficient gradient based optimization algorithm, which allows us to find optimal broadcast intensities for a variety of visibility shaping tasks in
a matter of milliseconds.
Finally, we experimented on a large real-world dataset ga\-thered from Twitter dataset, and show that our framework can consistently make broadcasters' posts more visible than 
alternatives.




\xhdr{Related work}
The work most closely related to ours is by Spasojevic et al.~\cite{spasojevic2015post}, who introduced the when-to-post 
pro\-blem. In their work, they first perform an empirical study on the best times to post in Twitter and Facebook by ana\-lyzing more 
than a billion messages and responses. 
Then, they design several heuristics to (independently) pinpoint at the times that elicited the greatest number of responses 
in a training set and then show that these times also lead to more responses in a held-out set.
In our work, we measure attention by means of visibility, a measure that is not confounded with the message content and can
be accurately evaluated on a held-out set, and then develop a convex optimization framework to design complete broadcasting 
strategies that are \emph{provably} optimal.

There have been an increasing number of empirical stu\-dies on understanding attention and information overload on social and information 
networks~\cite{backstrom2011center,hodas2012visibility,miritello2013limited,gomez14icwsm}.
The common\- theme is to investigate whether there is a limit on the amount of ties (\eg, friends, followees or phone contacts) people 
can maintain, how people distribute attention across them, and how attention influences the propagation of information.
In contrast, in this work, we focus on optimizing a social media user'{}s broadcasting strategy to capture the greatest attention from 
their followers.

Our work also relates to the influence maximization pro\-blem, extensively studied in recent years~\cite{RicDom02,KemKleTar03,CheWanWan2010,du13nips},
which aims to find a set of nodes in a social network whose initial adoptions of certain idea or product can trigger the largest expected number of follow-ups.
In this line of work, the goal is finding these influential users but not to find the best times for these users to broadcast their messages, which is 
our goal here.
Only very recently, Farajtabar et al.~\cite{shaping14nips} have developed a convex optimization framework to find broadcasting strategies, however, their focus
is on \-steering the overall activity in the network to a certain state by incentivizing a few influential users, in contrast, we focus on maximizing visibility 
as measured on a broadcaster'{}s audience'{}s feeds.

Finally, the framework of temporal point processes, which our work builds upon, has been increasingly used to model a wide range of
phenomena in social media and social networking sites, \eg, from social influence~\cite{shaping14nips}, network evolution~\cite{farajtabar2015coevolve}, 
opinion dynamics~\cite{de2015modeling} or product competition~\cite{competing15icdm}.

%% file: 020background.tex
A temporal point process is a stochastic process whose rea\-li\-za\-tion consists of a list of discrete events localized in time, $\cbr{t_i}$ with $t_i \in \RR^+$ and $i \in \ZZ^+$. Many 
different types of data produced in online social networks can be represented as temporal point processes, such as the times of tweets, retweets or likes in Twitter. A temporal 
point process can be equivalently re\-pre\-sen\-ted as a counting process, $N(t)$, which records the number of events before time $t$. Then, in a infinitesimally small time window 
$dt$ around time $t$, the number of observed event is 
\begin{align}
  \label{eq:count_change}
  d N(t) = \sum_{t_i \in \Hcal(t)} \delta(t-t_i) \, dt,  
\end{align}
and hence $N(t) = \int_0^t dN(s)$, where $\delta(t)$ is a Dirac delta func\-tion. It is often assumed that only one event can happen in a small window of size $dt$, and hence $dN(t) \in \cbr{0,1}$. 

An important way to characterize temporal point processes is via the intensity function --- the stochastic model for the time of the next event given all the times of pre\-vious events. The intensity function $\lambda(t)$ (intensity, for short) is the probability of observing an event in a small window $[t, t+dt)$, \ie,
\begin{align}
  \label{eq:intensity}
  \lambda(t)dt = \PP\cbr{\text{event in $[t, t+dt)$}}.
\end{align}
Based on the intensity, one can obtain the expectation of the number of events in the windows $[t,t+dt)$ 
and $[0,t)$ respectively as
\begin{equation}
  \EE[dN(t)] = \lambda(t)\, dt~\text{and}~
  \EE[N(t)] = \int_0^t \lambda(\tau)\, d \tau
\end{equation}
There is a wide variety of functional forms for the intensity $\lambda(t)$ in the growing literature on social activity modeling using\- point processes,
which are often designed to capture the phenomena of interests. 
For example, retweets have been modeled using multidimensional Hawkes processes~\cite{shaping14nips,zhao2015seismic}, new 
network links have been predicted using survival processes~\cite{vu2011continuous,farajtabar2015coevolve}, and daily and weekly variations on message broadcasting 
intensities have been captured using inhomogeneous Poisson processes~\cite{navaroli2015modeling}.

In this work, since we are interested on optimizing message broadcasting intensities, we use inhomogeneous Poisson processes, whose intensity is a time-varying 
function $\lambda(t) = g(t) \geqslant 0$. 

%
%

%% file: 030visibility.tex
In this section, we will present our model for the pos\-ting times of broadcasters and the feed story arrival times of followers using point processes parameterized by intensity functions. Based on these models, we will then define our visibility measure, and derive a novel link between the vi\-si\-bi\-li\-ty measure and the intensity functions of a broadcaster and her followers. 

\xhdr{Representation of broadcast and feed}
Given a directed social network $\Gcal = (\Vcal, \Ecal)$ with $m = |\Vcal|$ users, we assume that each user can be both broadcaster and follower. Then, we will use two sets of counting processes to modeling each 
user'{}s activity, the first set for the user'{}s broadcasting activity, and the second set for the user'{}s feed activity. 

More specifically, we represent the broadcasting times of the users as a set of counting processes denoted by a vector $\Nb(t)$, in which the $u$-th dimension, $N_u(t) \in \{0\} \cup \ZZ^+$, counts 
the number of messages user $u$ broadcasted up to but not including time $t$.
Then, we can characterize the message rate of these users using their corresponding intensities
\begin{equation} \label{eq:broadcast-intensity}
\EE[d\Nb(t)] = \lambdab(t) \, dt.
\end{equation}

Furthermore, given the adjacency matrix $\Ab \in \cbr{0, 1}^{m \times m}$ corresponding to the social network $\Gcal$, where $A_{uv} = 1$ indicates that $v$ follows $u$, and $A_{uv} = 0$ otherwise, 
we can represent the feed story arrival times of the users as a sum of the set of broadcasting counting processes. That is 
\begin{align}
  \Mb(t) = \Ab^{T} \Nb(t),
\end{align}
which essentially aggregates for each user the counting processes of the broadcasters followed by this user. 
Then, we can characterize the feed rates using intensity functions
\begin{equation} \label{eq:feed-intensity}
\EE[d\Mb(t)] = \gammab(t)  \, dt,
\end{equation}
where $\gammab(t) := \Ab^{T} \lambdab(t) = (\gamma_1, \ldots, \gamma_{m})^{T}$.

Finally, from the perspective of a pair of broadcaster (or user) $u$ and her follower $v$, it is useful to define the feed rate of $v$ due to 
other broadcasters (or users) followed by $v$ as
\begin{equation} \label{eq:feed-intensity-broadcaster}
\gamma_{v\setminus u}(t) := \gamma_v(t) - \lambda_u(t),
\end{equation}
where we assume $\gamma_{v\setminus u}(t) := 0$ if $v$ does not follow $u$, $A_{uv} = 0$.

\xhdr{Definition of Visibility} 
%
Consider a broadcaster $u$ and her follower $v$, and we note that $v$ may follow many other broadcasters other than $u$. Thus, at any time $t$, user $v$ may see stories originated from multiple broadcasters. We can model the times and origins of all these stories present in $v$'s current feed as a first-in-first-out (FIFO) queue\footnote{\scriptsize In this work, we assume the social network sorts stories in each user'{}s feed in inverse chronological order.} of pairs 
\begin{equation}
  \Hcal_v(t) := \left \{ (t_{(i)}, u_{(i)}) : \right.
  \left. t\geqslant t_{(1)}\geqslant\ldots\geqslant t_{(I-1)}\geqslant t_{(I)},\, u_{(i)} \in \Ncal^{-}(v) \right \}, \nonumber
\end{equation}
%
where $\cdot_{(i)}$ denotes the $i$-th element in the queue, $t_{(i)}$ is the time when $v$ receives a story from broadcaster $u_{(i)}$, $\Ncal^{-}(v)$ denotes the set of broadcasters followed by $v$, and $I$ is the length of the queue. The length $I$ accounts for the fact that online social platforms typically set a maximum number of stories that can be displayed in the feed, \eg, currently Twitter has $I=20$. 
%
%
The FIFO queue is to model the fact that when a new story arrives, the oldest story, $(t_{(I)},u_{(I)})$, at the bottom of the feed will be removed, and the ordering of the remaining stories will be shifted down by one slot,~\ie,
$$
  i+1 \leftarrow i, \quad \forall i = 1,\ldots, I-1
$$
and the newly arrived story will be appended to the be\-ginning of the queue as $t_{(1)}$ and appear at the top of the feed. For simplicity, we assume that the queue is always full at the time of modeling. 

In the list $\Hcal_v(t)$, we keep track of the rank $r_{uv}(t)$ of the most recent story posted by the broadcaster $u$ among all the stories received by user $v$ 
by time $t$, \ie,
\begin{equation}
r_{uv}(t) = \min \cbr{i~:~u_{(i)} = u}.
\end{equation}
Then, given an observation time window $[0, T]$, and a deterministic sequence of broadcasting events, we can define the deterministic visibility of broadcaster $u$ at $k$ with respect to follower $v$ as
\begin{equation} \label{eq:visibility}
\Tcal_{uv}(k) := \int_0^T \II[r_{uv}(t) \leqslant k] \, dt,
\end{equation}
which is the amount of times that at least one story from broadcaster $u$ is among the most recent $k$ stories in user $v$'s feed. 
%
%

Since the sequence of broadcasting events are generated from stochastic processes, we will consider the expected value of $\Tcal_{uv}(k)$ instead. If we first denote the probability that at least one story 
from broadcaster $u$ is among the $k$ most recent stories in follower $v$'s feed as
\begin{equation}
f_{uv}(t,k) = \PP\cbr{r_{uv}(t) \leqslant k}, 
\end{equation}
then the expected (or average) visibility $\Vcal(k)$ can be defined as
\begin{equation} \label{eq:average-visibility}
  \Vcal_{uv}(k) := \EE \sbr{\Tcal_{uv}(k)} = \int_0^T f_{uv}(t,k) \, dt,
\end{equation}
given the integral is well-defined. 
In some scenarios, one may like to favor some periods of times (\eg, times in which the follower is online), encode such preference
by means of a time significance function $s(t) \geqslant 0$ and consider $f_{uv}(t,k) s(t)$ instead of just $f_{uv}(t,k)$.

Note that the visibility $\Vcal_{uv}(k)$ is defined for a pair of broadcaster $u$ and her follower $v$ given $k$. We will focus our later exposition on a particular of $u$ and $v$, and omit the subscript $\cdot_{uv}$ and simply use notation such as $f_k(t)$, $\Vcal(k)$. However, we note that the computation of the visibility for a pair of users $u$ and $v$ may depend on the broadcast and feed intensities of all users in the network. 

\xhdr{Computation of Visibility}
In this section, we derive an expression for the average visibility, given by Eq.~\ref{eq:average-visibility}, using the broadcaster posting and follower feed representation,
given by Eqs.~\ref{eq:broadcast-intensity}-\ref{eq:feed-intensity-broadcaster}. This link is crucial for the convex visibility shaping framework in Section~\ref{sec:formulation}.

Given a broadcaster $u$ with $\lambda_u(t) = \lambda(t)$ and her follower $v$ with $\gamma_v(t) = \gamma(t)$
and $\gamma_{v\setminus u}(t) = \mu(t)$, we first compute the probability $f_1(t)$ that at least one message from the broadcaster is among the $k=1$ most recent ones 
received by $v$ at time $t$.
By definition, one can easily realize that $f_1(t)$ satisfies the following equation:
\begin{equation}
f_1(t+ dt) =   \underbrace{f_1(t) \, (1- \mu(t)  dt)}_{\text{\scriptsize 1. Remains the most recent}} +  \underbrace{ (1-f_1(t)) \, \lambda(t) dt,}_{\text{2. Becomes the most recent}}
\end{equation}
where each term models one of the two possible situations:
\begin{compactitem}
\item[1.] The most recent message received by follower $v$ by time $t$ was posted by broadcaster $u$ (\wpr~$f_1(t)$) and none of the other broadcasters that $v$ follows posts a message 
in $[t, t+dt]$  (\wpr~$1 - \mu(t) dt$).
\item[2.] The most recent message received by follower $v$ by time $t$ was posted by a different broadcaster (\wpr~$1-f_1(t)$) and 
broadcaster $u$ posts a message in $[t, t+dt]$ (\wpr~$\lambda(t) dt$) which becomes the most recent one.
\end{compactitem}
Then, by rearranging terms and letting $dt \to 0$, one finds that the probability satisfies the following differential equation:
%
%
\begin{align}
\label{eq:ode-for-f}
f_1'(t) = -(\mu(t)+\lambda(t)) f_1(t) +  \lambda(t).
\end{align}

We can proceed with the induction step for $f_k(t)$ with $k > 1$. In particular, by definition, $f_k(t)$  satisfies the following equation:
\begin{equation}
 f_k(t+dt) =  \underbrace{f_{k-1}(t)}_{\text{1. Was among $k$$-$$1$}} + \underbrace{(f_{k}(t)-f_{k-1}(t)) (1- \mu(t) dt)}_{\text{2. Remains on the $k$-th position}} 
 + \underbrace{(1-f_k(t)) \lambda(t) dt,}_{\text{3. Becomes the most recent}}
\end{equation}
where each term models one of the three possible situations:
\begin{compactitem}
\item[1.] The last message posted by broadcaster $u$ by time $t$ is among the most recent $k-1$ ones received by follower $v$ (\wpr~$f_{k-1}(t)$) and, independent of 
whether a message is posted by any other broadcaster or not, this message will remain among the most recent $k$ at $t+dt$.
\item[2.] The last message posted by broadcaster $u$ by time $t$ is the $k$-th one (\wpr~($f_{k}(t)-f_{k-1}(t))$) and none of the other broadcasters followed by $v$ posts a message 
in $[t, t+dt]$  (\wpr~$1 - \mu(t) dt$)

\item[3.] The last $k$ messages received by follower $v$ by time $t$ were posted by other broadcasters (\wpr~$1-f_k(t)$) and broadcaster $u$ posts a message in 
$[t, t+dt]$ (\wpr~$\lambda(t) dt$), becoming the most recent one.
\end{compactitem}
By rearranging terms and letting $dt \rightarrow 0$, we uncover a recursive relationship between $f_k(t)$ and $f_{k-1}(t)$, by means of the following 
differential equation:
\begin{align}
\label{eq:prob-diff-eq-top-k}
{f_k}'(t) = -(\mu(t)+\lambda(t)) f_k(t) + \mu(t) f_{k-1}(t) + \lambda(t),
\end{align}

Perhaps surprisingly, we can find a closed form ex\-pression for $f_k(t)$, given by the following Lemma (proven in 
the Appendix~\ref{sup:proof-lem-f_k}):
\begin{lemma} \label{lem:f_k}
Given a broadcaster with message intensity $\lambda(t)$ and one of her followers with feed message intensity due to other broadcasters $\mu(t)$. The probability $f_k(t)$ 
that at least one message from the broadcaster is among the $k$ most recent ones received by the follower at time $t$ can be uniquely computed as
\begin{align} \label{eq:f_k}
 f_k(t) =  \frac{\int_{0}^{t} \lambda(\tau)  e^{-\int_{\tau}^{t} \lambda(x) dx}\, \Gamma[k,\int_{\tau}^{t} \mu(x) dx] \, d\tau} {(k-1)!},
\end{align}
given the boundary conditions $f_1(0) = \ldots = f_k(0) = 0$ and the incomplete Gamma function defined as $\Gamma[k,x] = \int_x^\infty \tau^{k-1} e^{-\tau} d \tau$.
\end{lemma}
\begin{figure}[t]
        \centering
        \begin{tabular}{c}
       		 \includegraphics[width=0.43\textwidth]{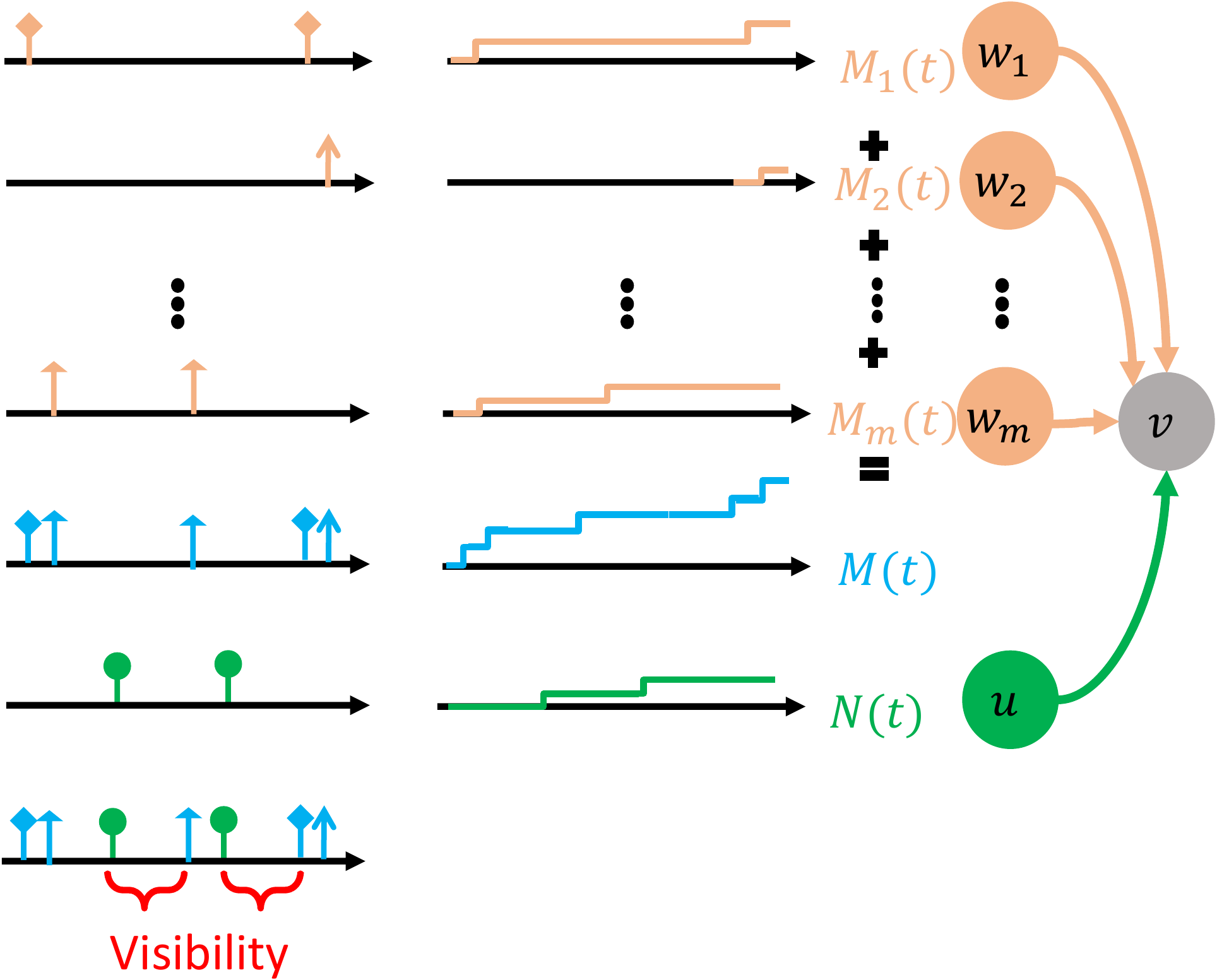}
        \end{tabular}
        \caption{The visibility shaping problem. A social media user $u$ broadcast $N_u(t)$ messages at a rate $\lambda_u(t)$. Her messages accumulate in each of her followers'{} feeds, 
        which receives $M_v(t)$ messages at a rate $\gamma_v(t) = \lambda_u(t) + \gamma_{v\backslash u}(t)$, where $\gamma_{v\backslash u}(t)$ denotes the message rate due to other 
        broadcasters $v$ follows. 
        For each follower, the average visibility of a user u'{}s messages is defined as the time that a post from user $u$ is among the last $k$ stories the follower received. In the visibility 
        shaping problem, the goal is to optimize $\lambda_u(t)$ to steer visibility.}
        \label{fig:formulation}
\end{figure}

%% file: 040concavity.tex
Once we have a formula that allows us to compute the ave\-rage visibility given \emph{any} arbitrary intensities for the broadcasters,
we will now show that, remarkably, the average visibility is concave in the space of smooth intensity functions.
Moreover, we will also show that the average visibility is concave with respect to the parameters of piecewise constant functions, which we 
will use in our experiments.

\xhdr{Smooth intensity functions}
In this section, we assume that the message intensity of the broadcaster belongs to the space $\Hcal$ of all smooth functions.
Before we proceed, we need the following definition:
%

\begin{definition}
Given the space $\Hcal$ of all smooth functions, a functional $J:\Hcal \rightarrow \RR$ is concave if for every 
$g_1 ,g_2 \in \Hcal$ and $ 0 < \alpha < 1$:
\begin{equation} \label{eq:convexity-def}
J[\alpha g_1 + (1-\alpha) g_2]  \ge \alpha J[g_1] + (1-\alpha) J[g_2].
\end{equation}
A functional $J$ is convex if $-J$ is concave.
\end{definition}

It readily follows that the probability $f_k(t)$, given by Eq.~\ref{eq:f_k}, is a functional with  $\lambda(\cdot)$ as input. Moreover, the following two
theorems, proven in Appendices~\ref{sup:concave-f-k} and~\ref{sup:concave-expected-T}, establish the concavity of $f_k(t)$ and $\Vcal(k)$ with respect to 
$\lambda(\cdot)$.

\begin{theorem}
\label{theo:concave-f-k}
Given a broadcaster with message intensity $\lambda(t)$ and one of her followers with feed message intensity due to other broadcasters $\mu(t)$. 
The probability $f_k(t)$ that at least one message from the broadcaster is among the $k$ most recent ones received by the follower at time $t$, given
by Eq.~\ref{eq:f_k}, is concave with respect to $\lambda(\cdot)$.
\end{theorem}
\begin{theorem}
\label{theo:concave-expected-T}
Given a broadcaster with message intensity $\lambda(t)$ and one of her followers with feed message intensity due to other broadcasters $\mu(t)$. 
The visibility $\Vcal(k)$, given by Eq.~\ref{eq:average-visibility}, is concave with respect to $\lambda(\cdot)$.
\end{theorem}

Given the above results, one could think of finding the optimal (general) message intensity $\lambda(t)$ that maximize (a function of) the average 
visibilities across a broadcaster'{}s follo\-wers. 
However, in practical applications, this may be inefficient and undesirable, instead, one may focus on a simpler parametrized family of intensities, 
such as piecewise constant intensity functions, which will be easier to optimize and fit using real data.
To this aim, next, we prove that the average visibility is also concave on the parameters defining piecewise constant intensity functions.

\xhdr{Piecewise constant intensity functions}
In this section, we assume that the message intensity $\lambda(t)$ of the broadcaster belongs to the space of piecewise constant functions 
$\lambda : [0, T] \rightarrow \RR$, denoted by $\Gcal$, which we parametrized as follows:
%
%
%
\begin{equation} \label{eq:piecewise-constant-function}
\lambda(t) = \sum_{m=1}^M a_m \II(\tau_{m-1} \le t < \tau_m),
\end{equation}
where $a_m \geq 0$, $M$ is the number of pieces, $\tau_i - \tau_{i-1} = T/M = \Delta$ and $\tau_0=0$.

As the reader may have noticed, the results from the previous section are not readily usable since Lemma~\ref{lem:f_k} requires the intensity functions 
to be smooth.
However, we will now show that, for every function $\lambda(t) \in \Gcal$, there is a sequence of smooth functions $\lambda_n(t) \in \Hcal$ such that 
$\lim_{n \rightarrow \infty} \lambda_n(t) = \lambda(t)$ and, this will sufficient to prove concavity.
Before we proceed, we need the following definition:
\begin{definition}
A functional $J: \Gcal \to \Hcal$ is said to be continuous at $\lambda(\cdot) \in \Hcal$ if for every $\epsilon > 0 $, there is a $\delta > 0$ such that
\begin{align}
| J[\lambda] - J[\lambdahat] | < \epsilon
\end{align}
provided that $|| \lambda - \lambdahat || < \delta$, where $||\cdot||$ is a norm in $\Hcal$. 
\end{definition}

It readily follows that the probability $f_k$ is a continuous functional on $\Hcal$. Moreover, we need the following lemma (proven in Appendix~\ref{sup:lem-limit}) to prove the concavity:
%
%
\begin{lemma}
\label{lem:sequence}
For every $\lambda(t) \in \Gcal$, there is a sequence of smooth functions $\lambda_n(t) \in \Hcal$ where $\lim_{n \to \infty} \lambda_n(t) = \lambda(t)$. 
\end{lemma}

Using Lemma~\ref{lem:sequence}, for any $\lambda(t) \in \Gcal$, it follows that
\begin{align}
f_k(\lambda(\cdot)) = \lim_{n \to \infty} f_k(\lambda_n(t))
\end{align} 
where $\lambda_n(t)$ is a sequence of smooth functions such that $\lim_{n \rightarrow \infty} \lambda_n(t) =  \lambda(t)$. 
As a consequence, we can establish the concavity of $f_k(t)$ and $\Vcal(k)$ with respect to $a_1, \ldots, a_M$ with the
following Theorem (proven in Appendix~\ref{app:piecewise-concave}):
\begin{theorem} \label{theo:piecewise-concave}
$f_k$ and $\Vcal(k)$ are concave functionals in the space of piecewise constant functions $\Gcal$.
\end{theorem}

\begin{corollary}
If we represent $\lambda \in \Gcal$ using Eq.~\ref{eq:piecewise-constant-function}, $f_k(t)$ and $\Vcal(k)$ are concave with respect to 
$a_1, \ldots, a_m$.
\end{corollary}

%% file: 050formulation.tex
Given the concavity of the average visibility, we now propose a convex optimization framework for a variety of vi\-si\-bi\-li\-ty shaping tasks. 
In all these tasks, our goal is to find the optimal message intensity $\lambda_{u}(t)$ for broadcaster $u$ that maximizes a particular nondecreasing concave utility function $U(\Vcalb_u(k))$ of the average visibility of broadcaster $u$ in all her followers within a time window $[0, T]$, \ie,
\begin{equation}
	\label{eq:generalized-activity-maximization}
	\begin{array}{ll}
		\mbox{maximize}_{\lambda_u(t)} & U(\Vcalb_u(k))  \\
		\mbox{subject to} 
		& \lambda_u(t) \ge 0  \quad  t \in [0,T] \\
		&  \int_0^T \lambda(t) \, dt \le C,
	\end{array}
\end{equation}
where $\Vcalb_u(k) = (\Vcal_{uv}(k))_{v \in \Ncal(u)^{+}}$, $\Ncal(u)^{+}$ denotes the broadcaster $u$'s followers, $\Vcal_{uv}(k)$ denotes the average visibility in follower $v$, 
the first constraint asserts the intensity function remains positive, 
and the second limits the average number of messages broadcasted within $[0, T]$ to be no more than $C$.
%

We next discuss two instances of the general framework, which achieve different goals (their constraints remain the same and hence omitted). More generally, the 
flexibility of our framework allows to use any nondecreasing concave utility function. 

\xhdr{Average Visibility Maximization (AVM)}
The goal here is to maximize the sum of the visibility for all the broadcaster'{}s followers, $\ie$,
\begin{align} \label{eq:avm-formulation}
\mbox{maximize}_{\lambda_u(t)}  \quad  \sum_{v \in \Ncal(u)^{+}}{ \Vcal_{uv}(k) }
\end{align}

%
%
%

\xhdr{Minimax Visibility Maximization (MVM)}
Suppose our goal is instead to keep the visibility in the $n$ followers with the smallest visibility value above a certain minimum level, or, 
alternatively make the average visibility across the $n$ followers with the smallest visibility as high as possible. Then, we can perform the 
following \emph{minimax visibility maximization} task
\begin{align} \label{eq:mvm-formulation}
 \mbox{maximize}_{\lambda_u(t)} \sum_{i=1}^{n} \Vcal_{uv_{[i]}}(k),
\end{align}
where $\Vcal_{uv_{[i]}}(k)$ denotes the average visibility in the follower with the $i$-th smallest visibility among all the
broadcaster'{}s followers.

%
%
%
%
%
%
%

%% file: 060algorithm.tex
\begin{algorithm}[t]
\caption{Projected Gradient Descent for Visibility Shaping}
\label{algorithm1}
Initialize $\cb$\;
\Repeat{convergence}
{
   1- Project $\cb$ into the polytope $\cb \geqslant 0$, $\cb^{\top} \one \leqslant C$\;
   2- Find the gradient $\gb(\cb)$\;
   3- Update $\cb$ using the gradient $\gb(\cb)$\;
}
\end{algorithm}

To solve the visibility shaping problems defined above, we need to be able to (efficiently) evaluate the probability function $f_k$ and visibility $\Vcal(k)$. However, a direct evaluation 
by means of Eqs.~\ref{eq:f_k} and~\ref{eq:average-visibility} seem difficult. 
Here, we present an alternative representation of the probability function $f_k$ and the visibility $\Vcal(k)$ for piece-wise constant intensity functions, which allow us to compute
these quantities very efficiently. Based on this result, we present an efficient gradient based algorithm to find the optimum intensity.

Assume the broadcaster'{}s message intensity $\lambda(t)$ and the follower'{}s feed message intensity due to other broadcasters $\mu(t)$ adopt the 
following form:
\vspace{-2mm}
\begin{equation} 
\lambda(t) = \sum_{m=1}^M c_m \II(\tau_{m-1} \le t < \tau_m)~\text{and}~\mu(t) = \sum_{m=1}^M b_m \II(\tau_{m-1} \le t < \tau_m). \nonumber
\end{equation}
\vspace{-2mm}
Then, each piece $m$ in the above intensities satisfies the recurrence relation given by Eq.~\ref{eq:prob-diff-eq-top-k}, which we rewrite as
\begin{equation}
f'_k(t) + (b_m+c_m) f_k(t)  = c_m + b_m f_{k-1}(t), \nonumber
\end{equation}
and one can easily prove by induction that, in general, the solution of the above differential equation for each time interval $\tau_{m-1} \leq t < \tau_m$ is given by
\begin{equation} \label{eqn-main-rec}
f_k(t) = e^{-(b_m+c_m)t}(\alpha_{k-1,k} t^{k-1}+\cdots+ \alpha_{0,k}) + \beta_k
\end{equation}
where $\alpha_{i,k} = \frac{b^i}{i!}(h_{k-i} - \beta_{k-i})$, 
$\beta_i = 1 - \left(\frac{b_m}{b_m+c_m}\right)^i$, and $h_i$ is the probability $f_i(\tau_{i-1})$ at the beginning of time interval. Such representation allows for an
efficient evaluation of $f_k(t)$.

Next, we also need to compute the integral of $f_k(t)$ to efficiently compute the visibility $\Vcal(k)$. Without loss of generality, we represent the time for each piece 
in a normalized time window $[0, 1]$. Then, the integral of $f_k(t)$ can be written as follows:
\begin{equation} 
\Vcal(k) = \int_0^1  f_k(t)\,dt 
=  \beta_k + \sum_{i=0}^{k-1} \alpha_{i, k}\int_0^1 e^{-(b_m+c_m)t}t^i\,dt
= \beta_k + \sum_{i=0}^{k-1} \frac{\alpha_{i, k}}{(b_m+c_m)^{i+1}}[i! - \Gamma(i+1, b_m+c_m)]  \label{eqn-last}
\end{equation}
where note that the last term is efficiently computable since, for integer values of $n$, the incomplete Gamma function $\Gamma(n, x) = (n-1)!\,e^{-x}\sum_{i=0}^{n-1} \frac{x^i}{i!}$.

Given Eq.~\ref{eqn-last}, we can now easily compute the gradient of the visibility $\Vcal(k)$, which we can then use to design an effi\-cient gradient based algorithm.
For brevity, we just show the gradient for $k=1$. Let $\cb = (c_1, \ldots, c_M)$, and $\yb = (y_0,\ldots, y_{M-1}, y_M)$ be the values of $f_k(t)$ at the beginning of each time 
interval, then,
\begin{equation}
\frac{\partial \Vcal(1)}{\partial c_i} = \frac{1}{(b_i+c_i)^2}\left(-\frac{\partial y_i}{\partial c_i}(b_i+c_i) + (y_i - y_{i-1}) + b_i\right) 
+ \sum_{m=i+1}^M\frac{1}{b_m+c_m}\left[\frac{\partial y_{m-1}}{\partial c_i} - \frac{\partial y_m}{\partial c_i}\right]. \nonumber
\end{equation}
where we can easily compute ${\partial y_j}/{\partial c_m}$ recursively as
\begin{align}
\frac{ b_m } {(b_m + c_m) ^ 2 } - \left(y_{m - 1} - \frac{c_m} { (b_m + c_m) } - \frac{ b_m } {(b_m + c_m) ^ 2 }\right) e^{-(b_m + c_m)}, \nonumber
\end{align}
if $j = m$, and $e^{-(b_j + c_j)} \frac{\partial y_{j-1}}{\partial c_m}$, if $j > m$.

Once we have an efficient way to compute the visibility $\Vcal(k)$ and its gradient, we can readily design a projected gradient descent algorithm to find the optimal message intensity $\lambda_u(t)$
in the visibility shaping problems described in Section~\ref{sec:formulation}. 
Note that, since our optimization problems are convex, there is a unique optimum and convergence is guaranteed. Moreover, for the projection step, we solve a quadratic program, minimizing the 
distance to the feasible polytope. 
Algorithm~\ref{algorithm1} summarizes the overall algorithm.

%% file: 070experiments.tex
%
\begin{figure*}[t]
\centering
\begin{tabular}{c c c c}
       		 \hspace{-3mm}\includegraphics[width=0.25\textwidth]{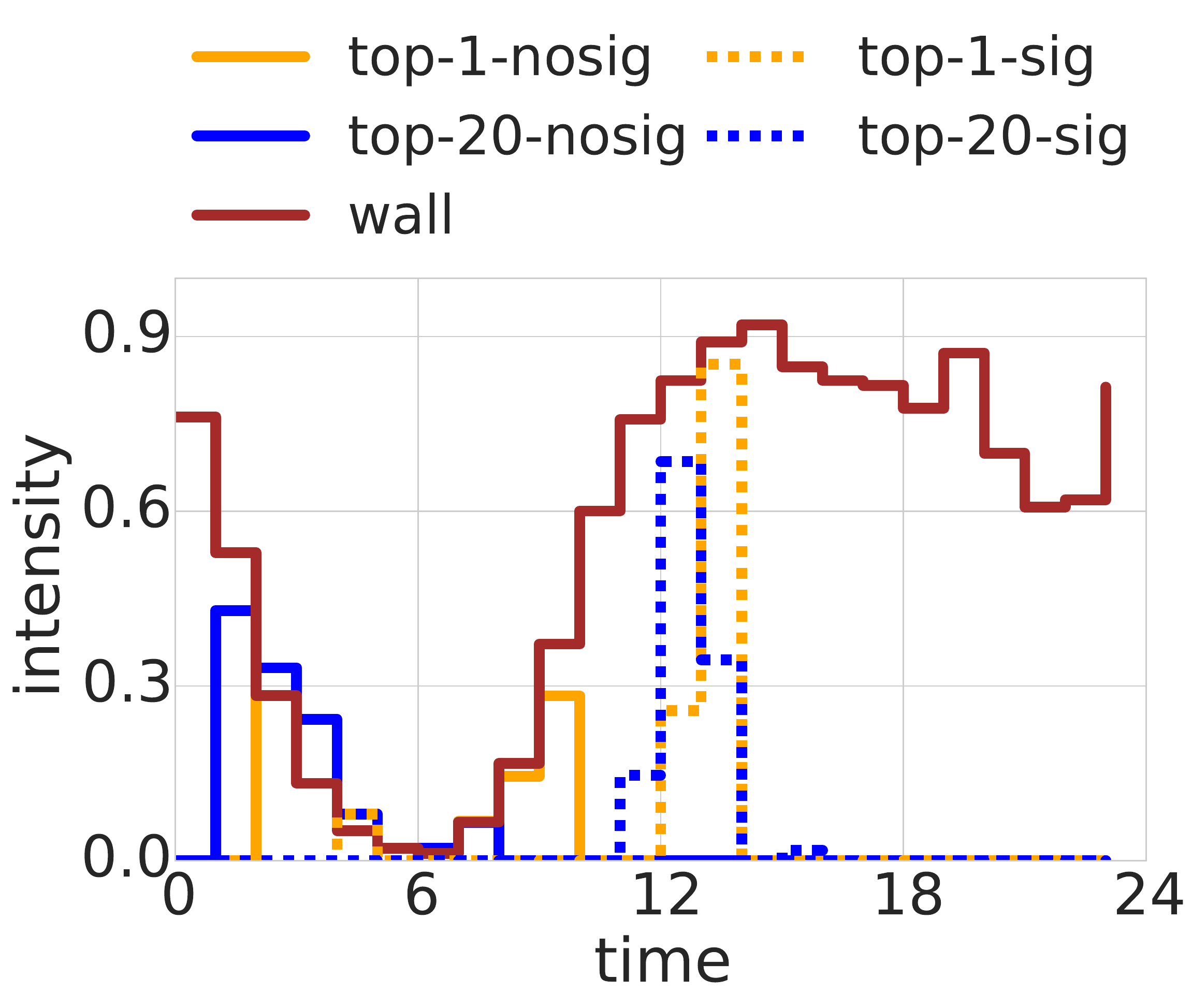} &
		 \hspace{-3mm}\includegraphics[width=0.24\textwidth]{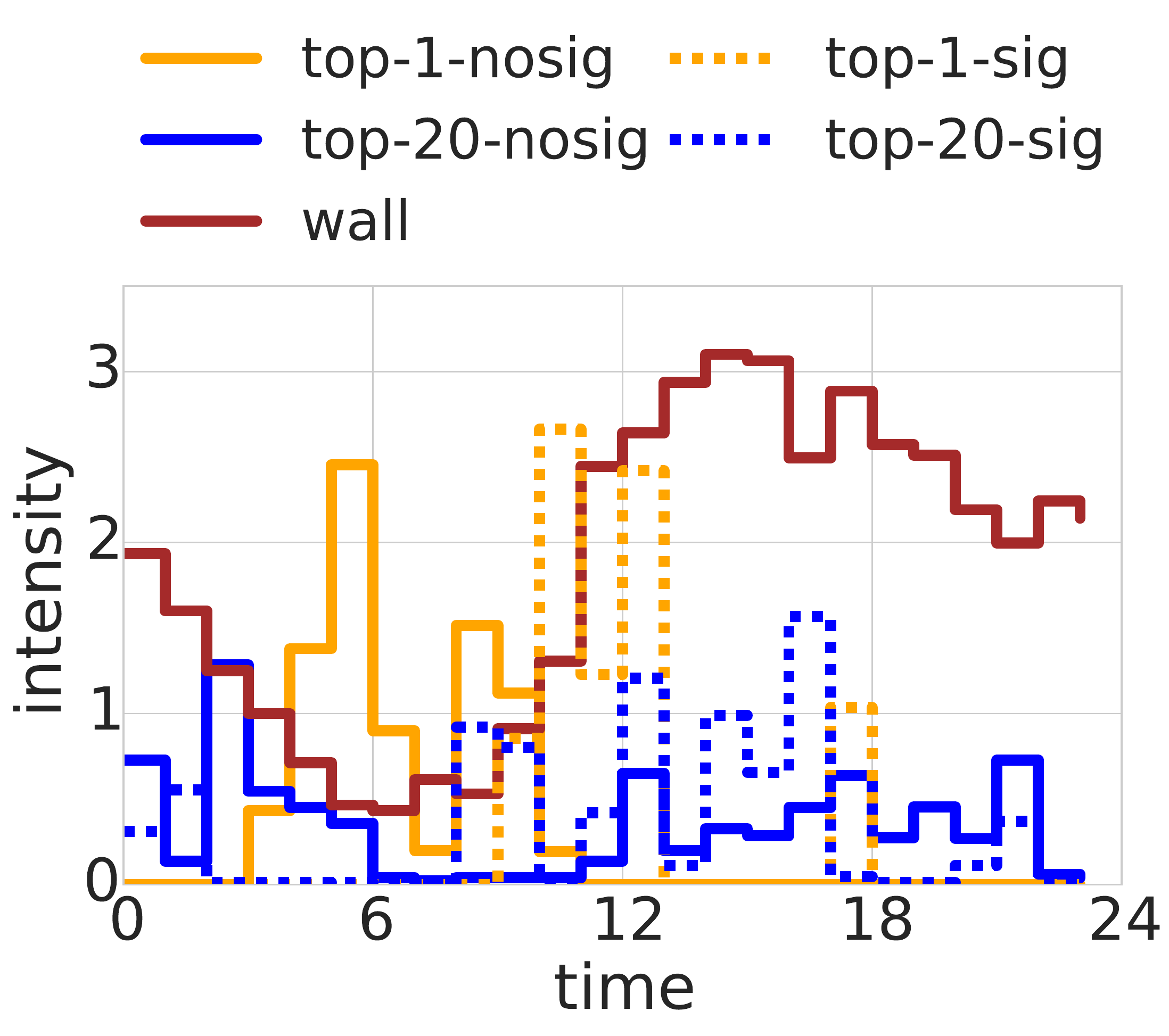} &
		 \hspace{-3mm}\includegraphics[width=0.24\textwidth]{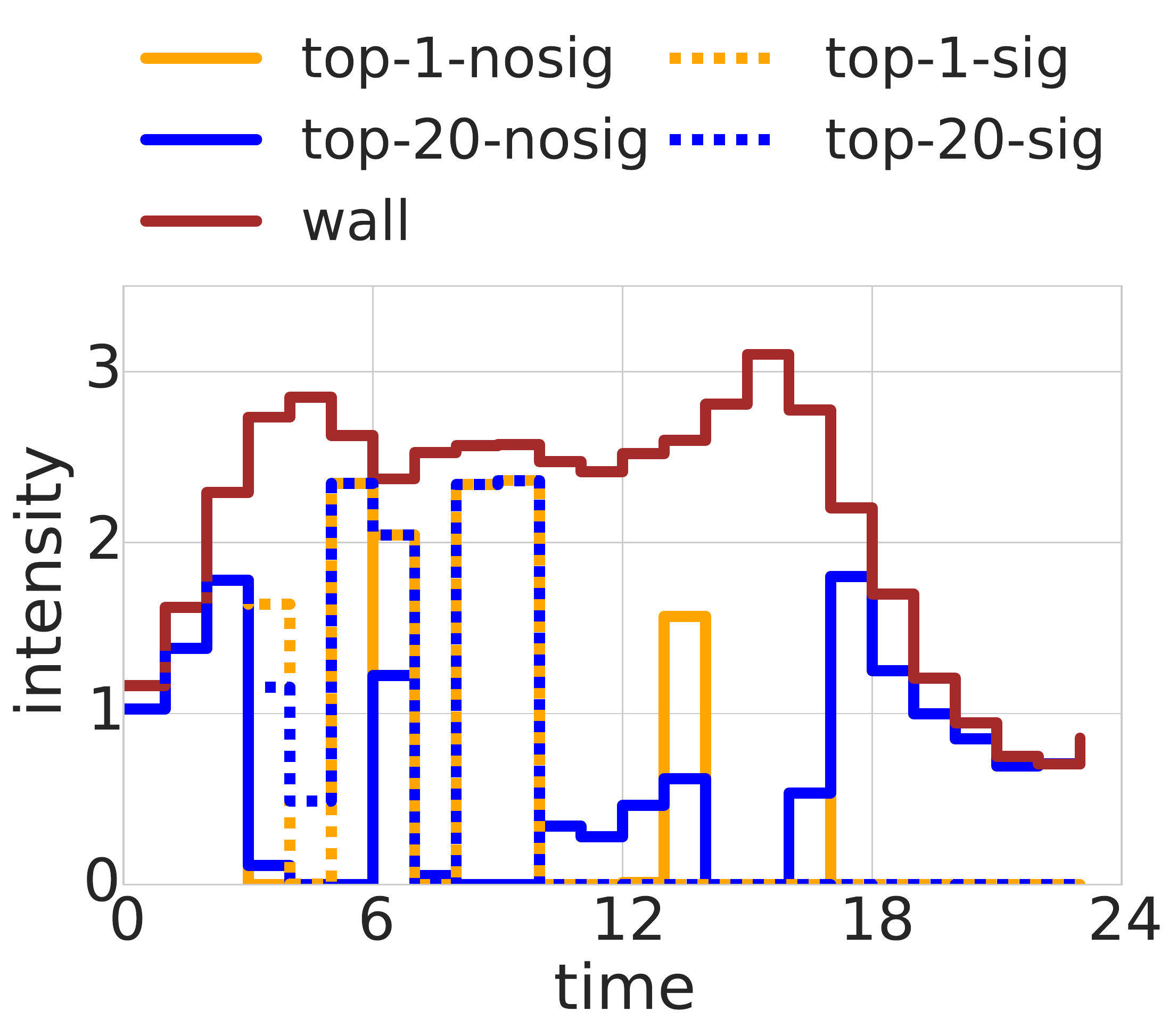} &
		 \hspace{-3mm}\includegraphics[width=0.25\textwidth]{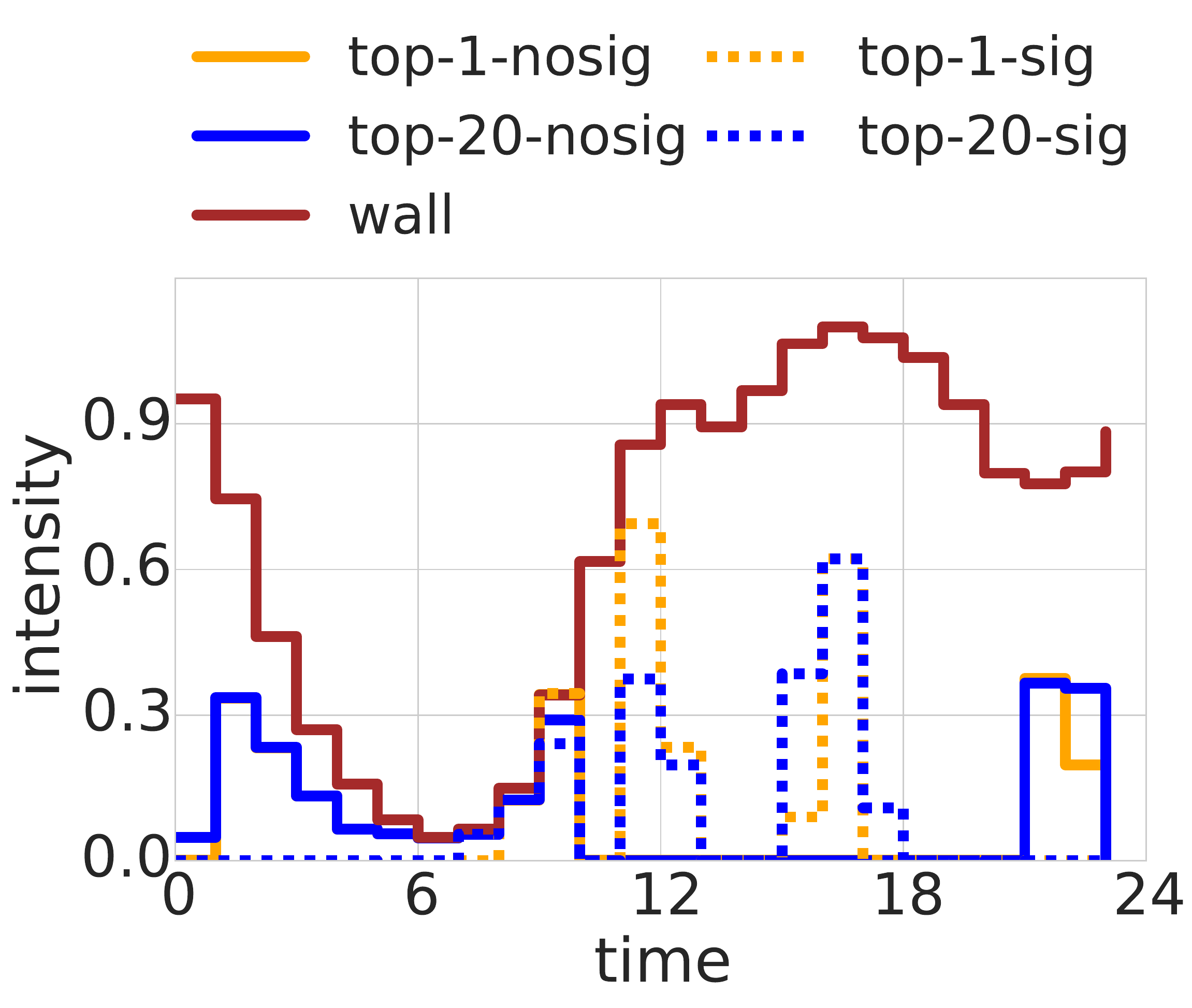} \\
       		 \hspace{-3mm}\includegraphics[width=0.25\textwidth]{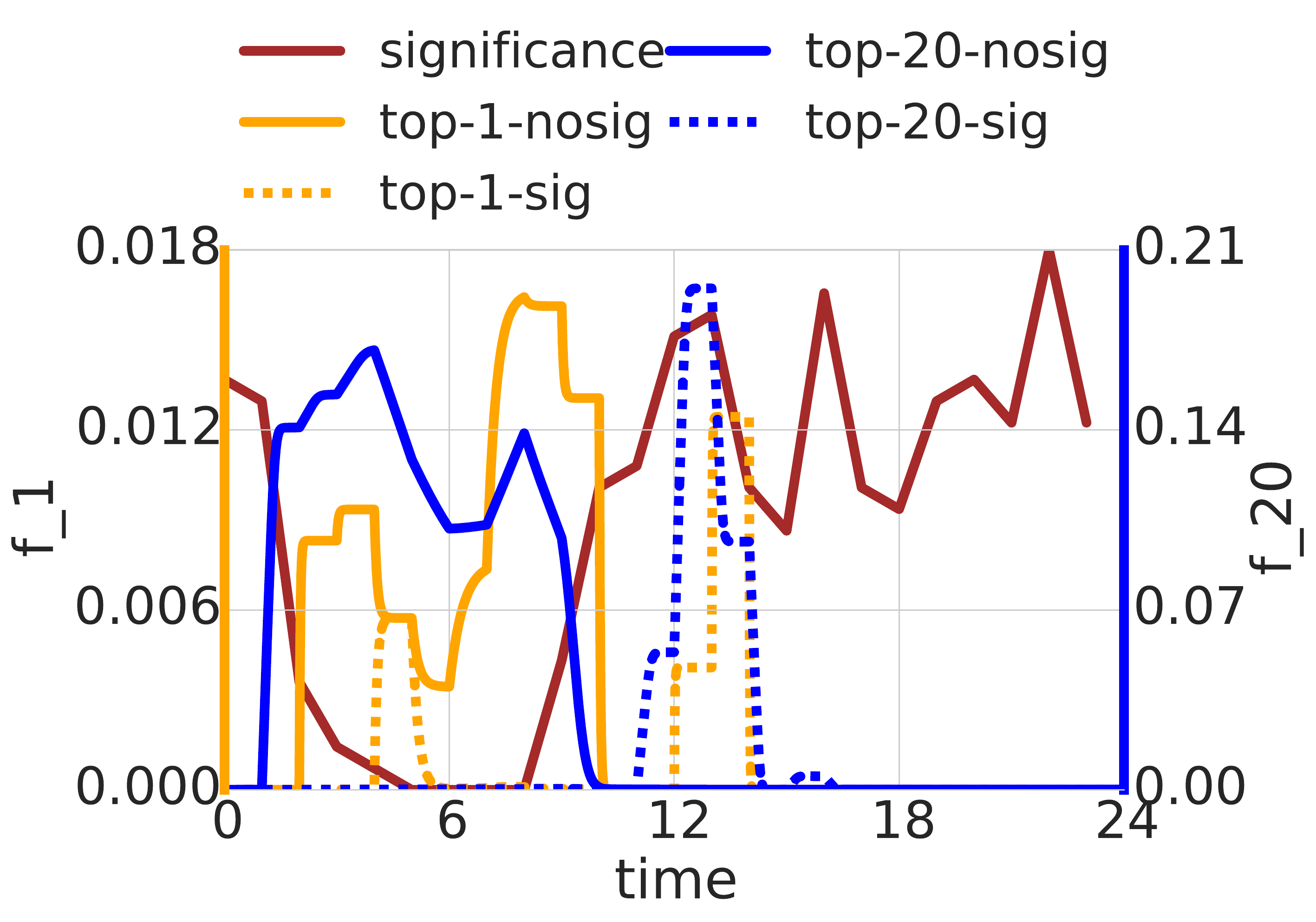} &
		 \hspace{-3mm}\includegraphics[width=0.25\textwidth]{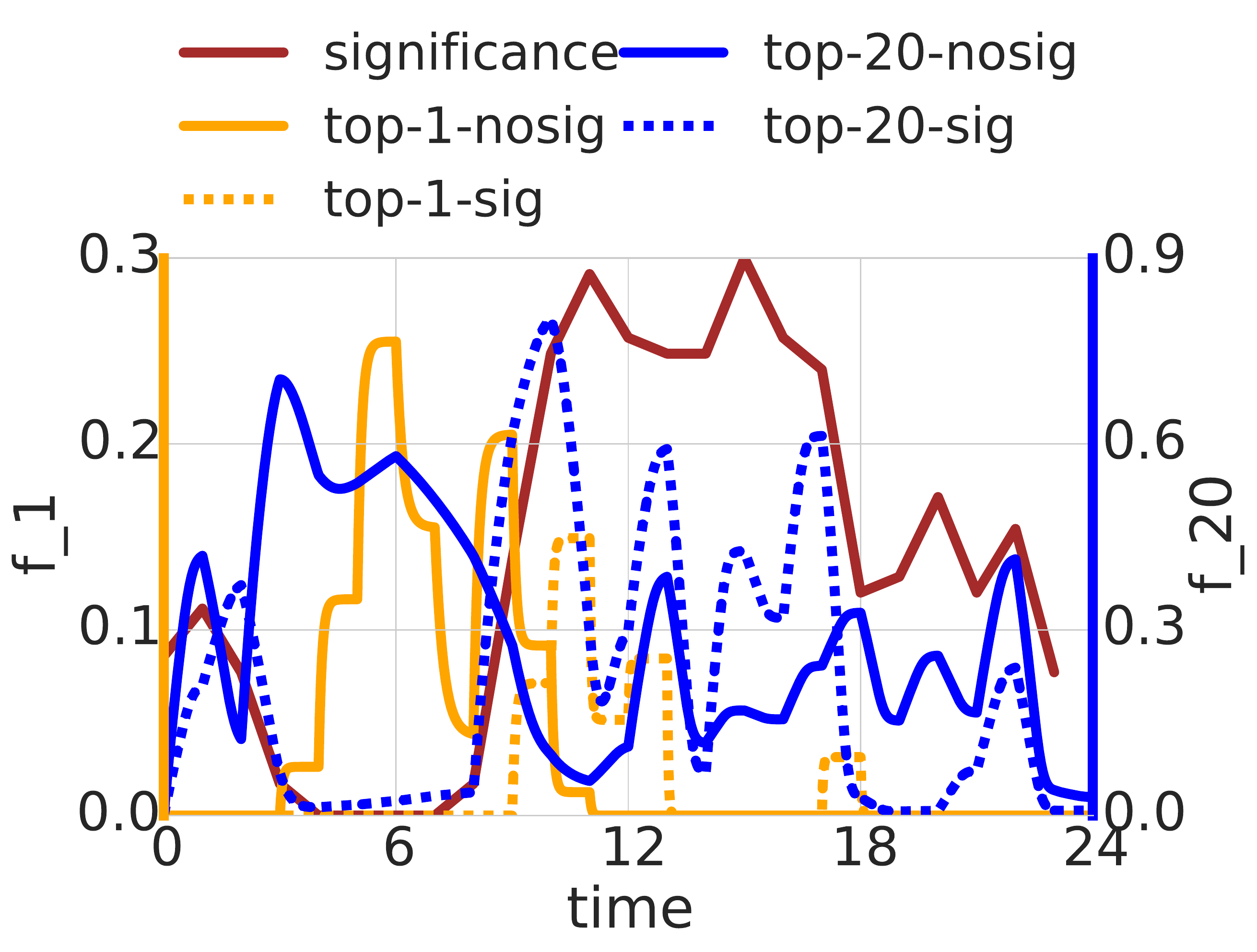} &
		 \hspace{-3mm}\includegraphics[width=0.25\textwidth]{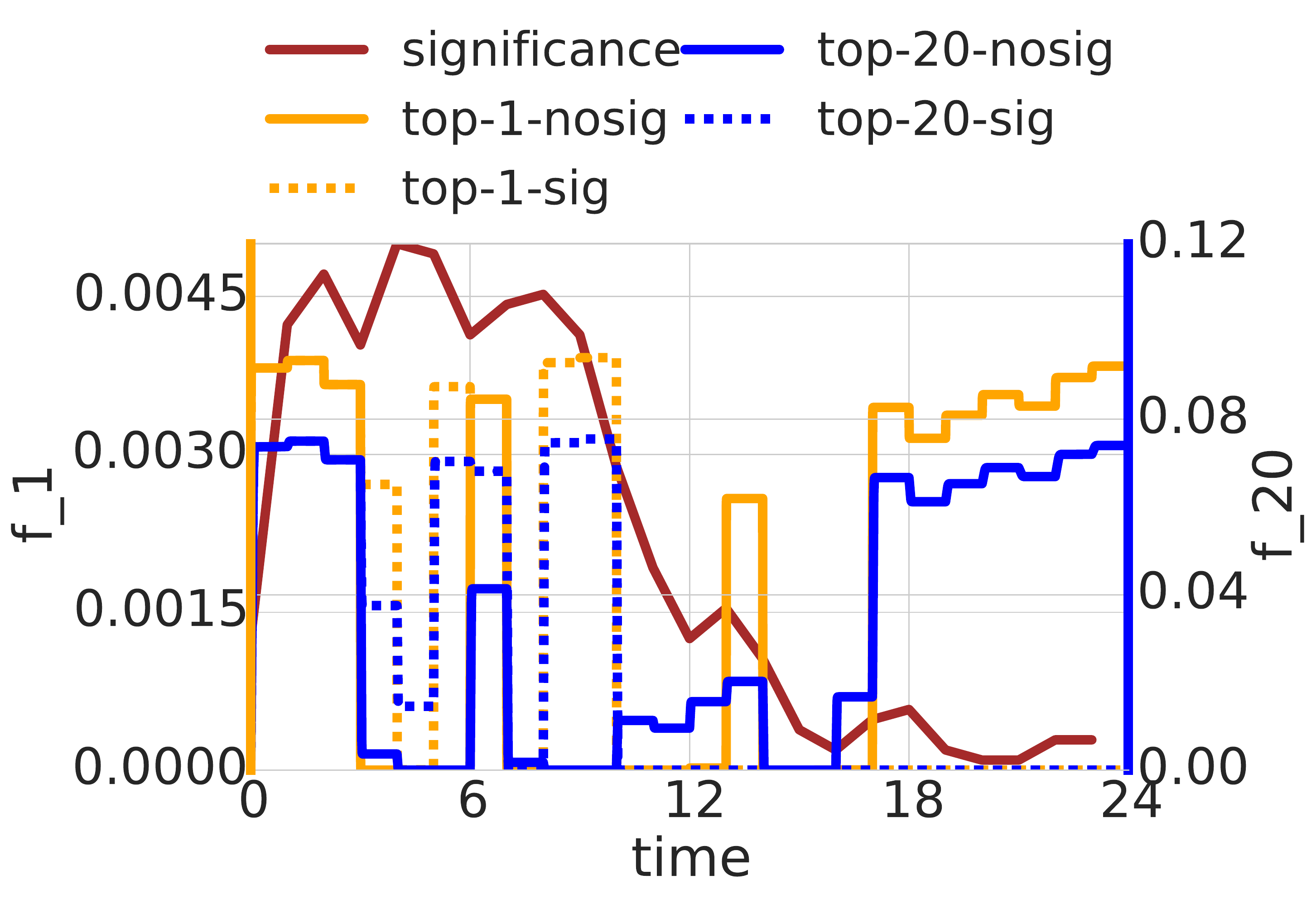} &
		 \hspace{-3mm}\includegraphics[width=0.25\textwidth]{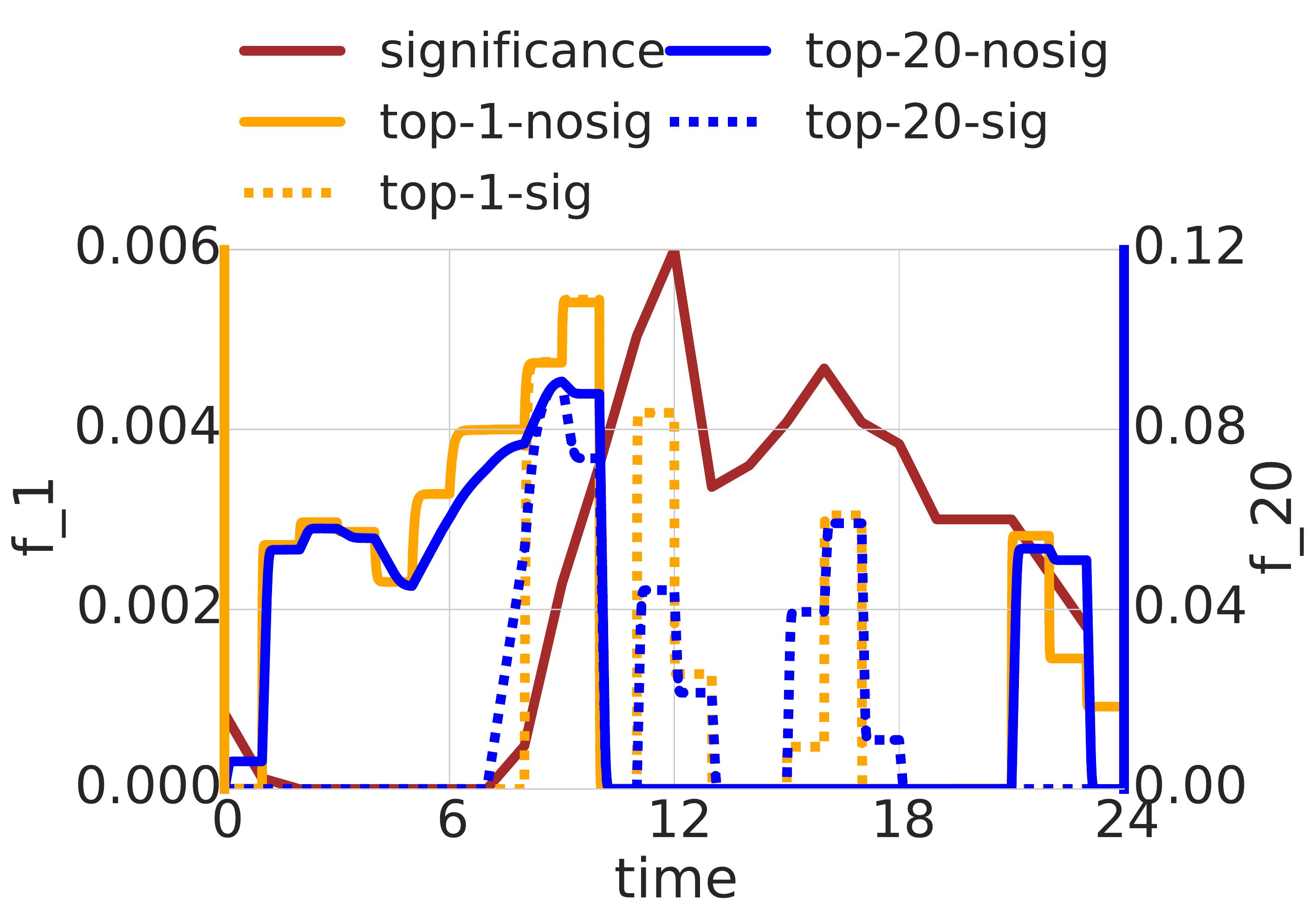} \\
      		 \hspace{-3mm}\includegraphics[width=0.25\textwidth]{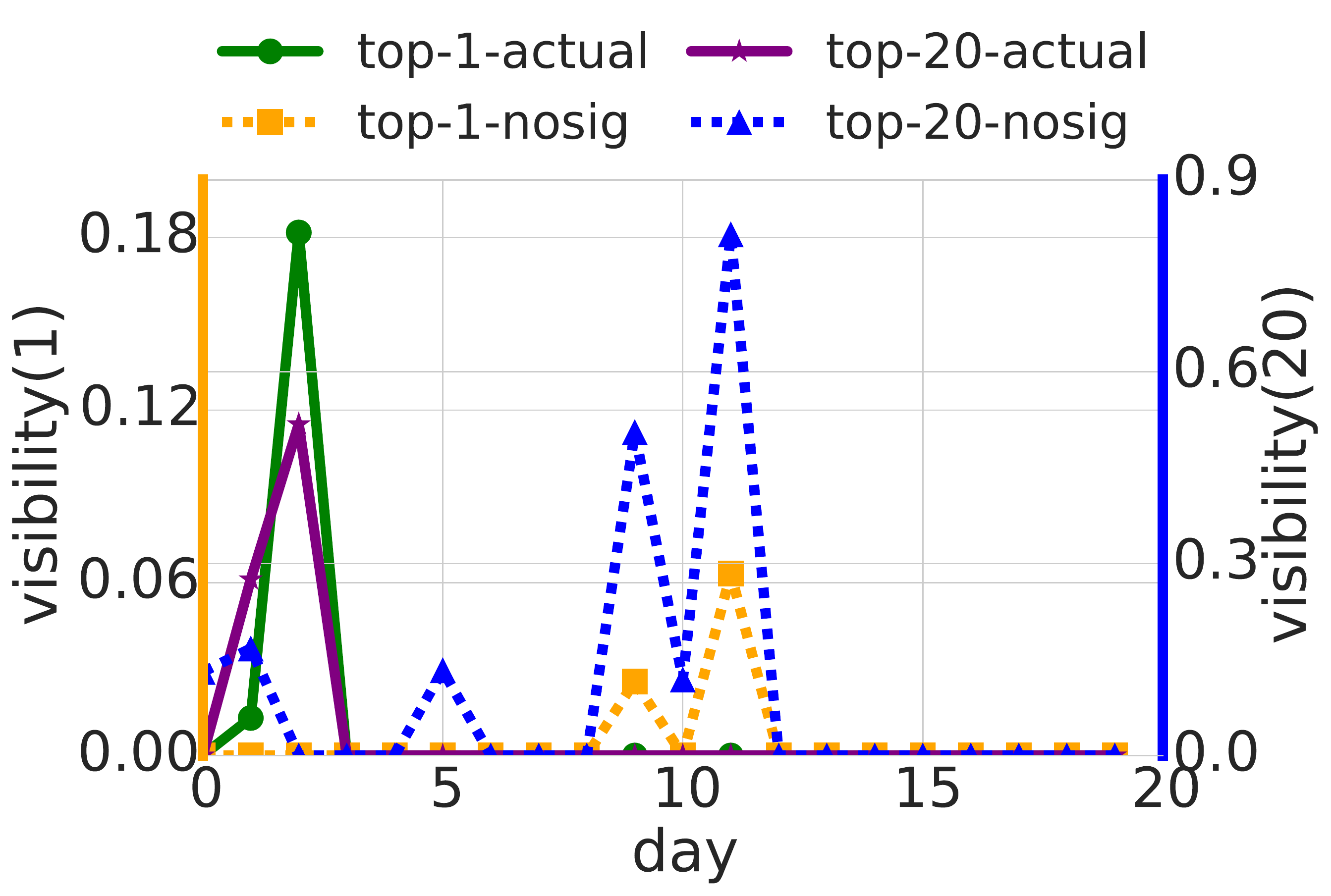} &
		 \hspace{-3mm}\includegraphics[width=0.25\textwidth]{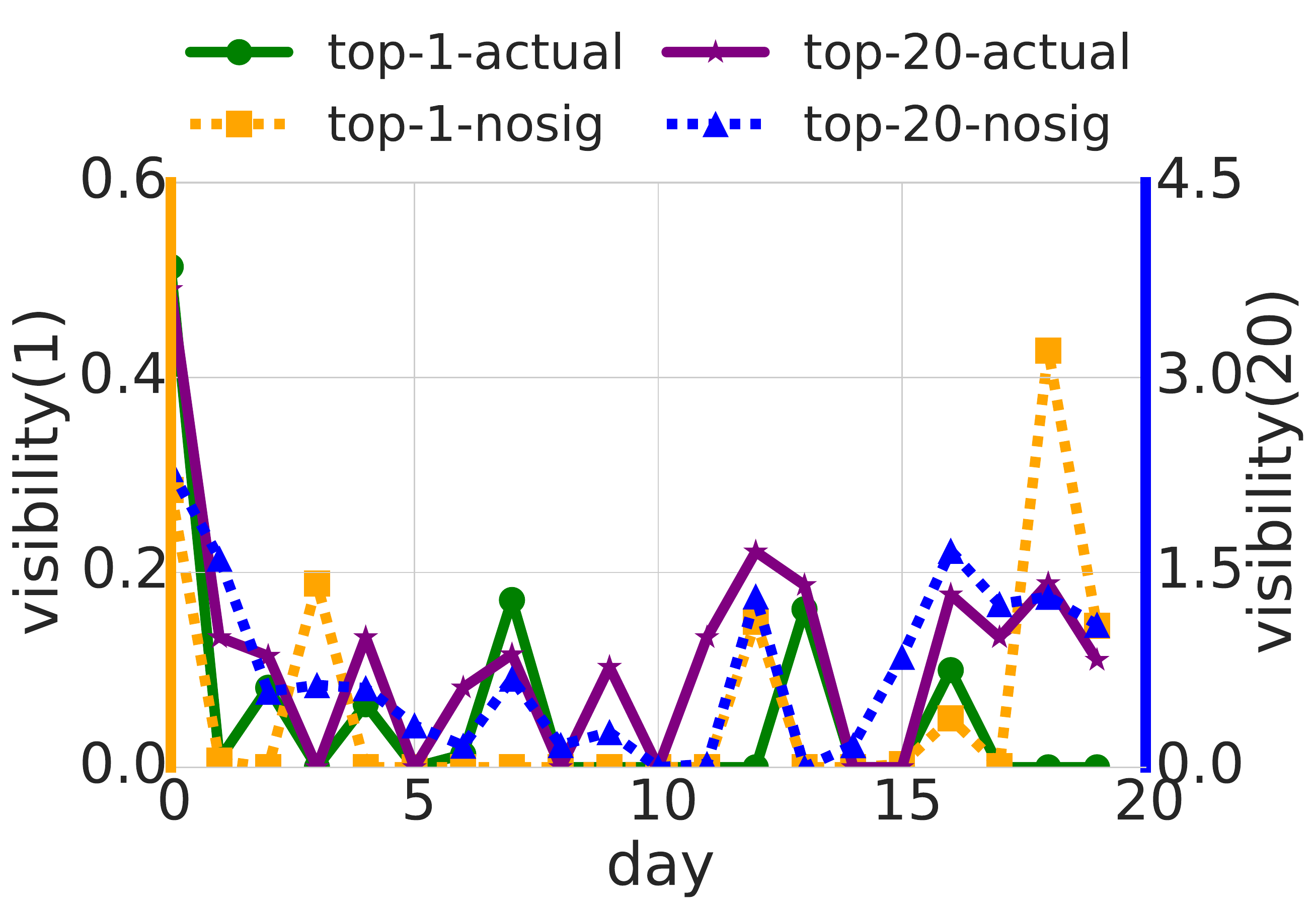} &
		 \hspace{-3mm}\includegraphics[width=0.25\textwidth]{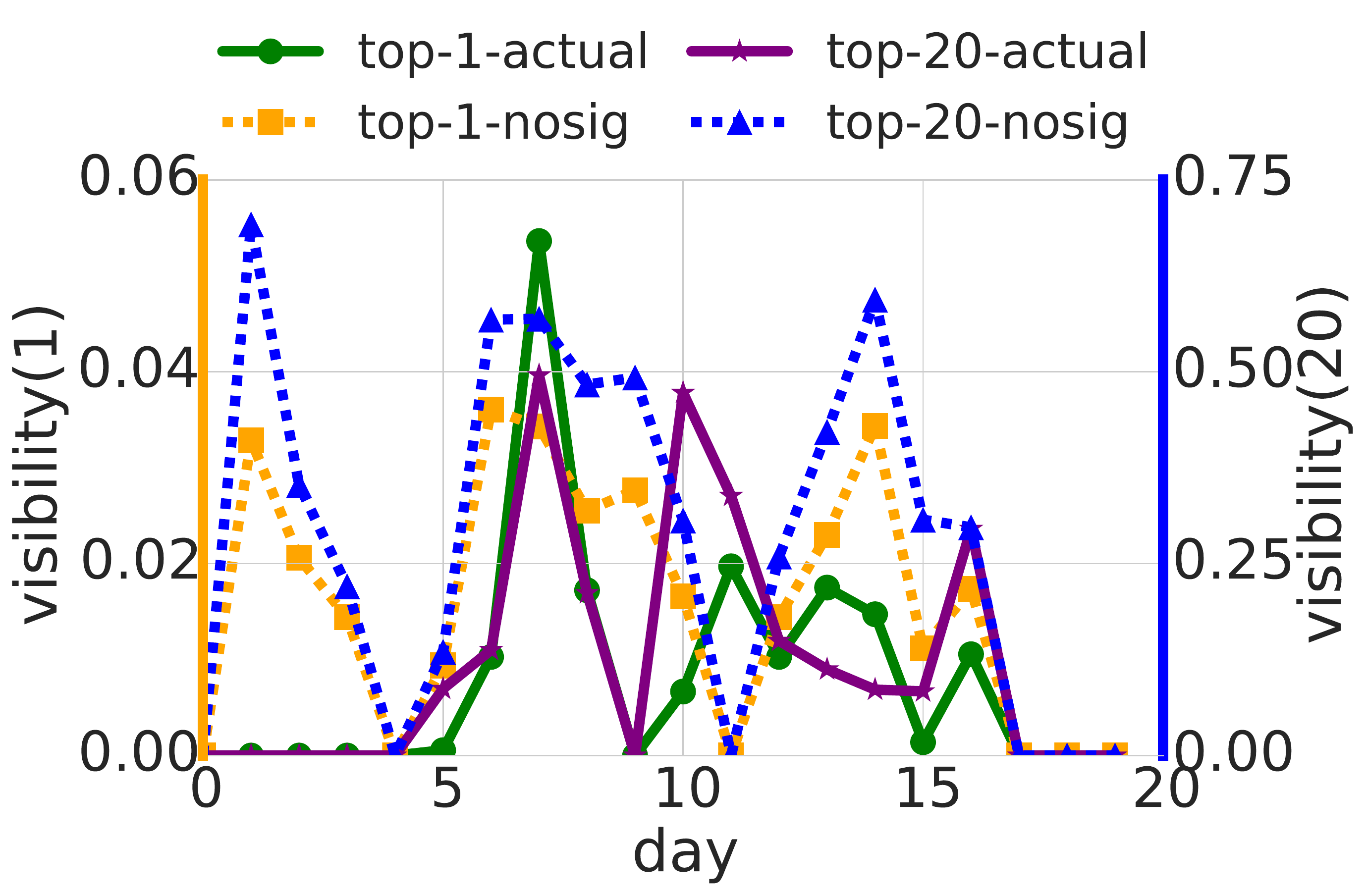} &
		 \hspace{-3mm}\includegraphics[width=0.25\textwidth]{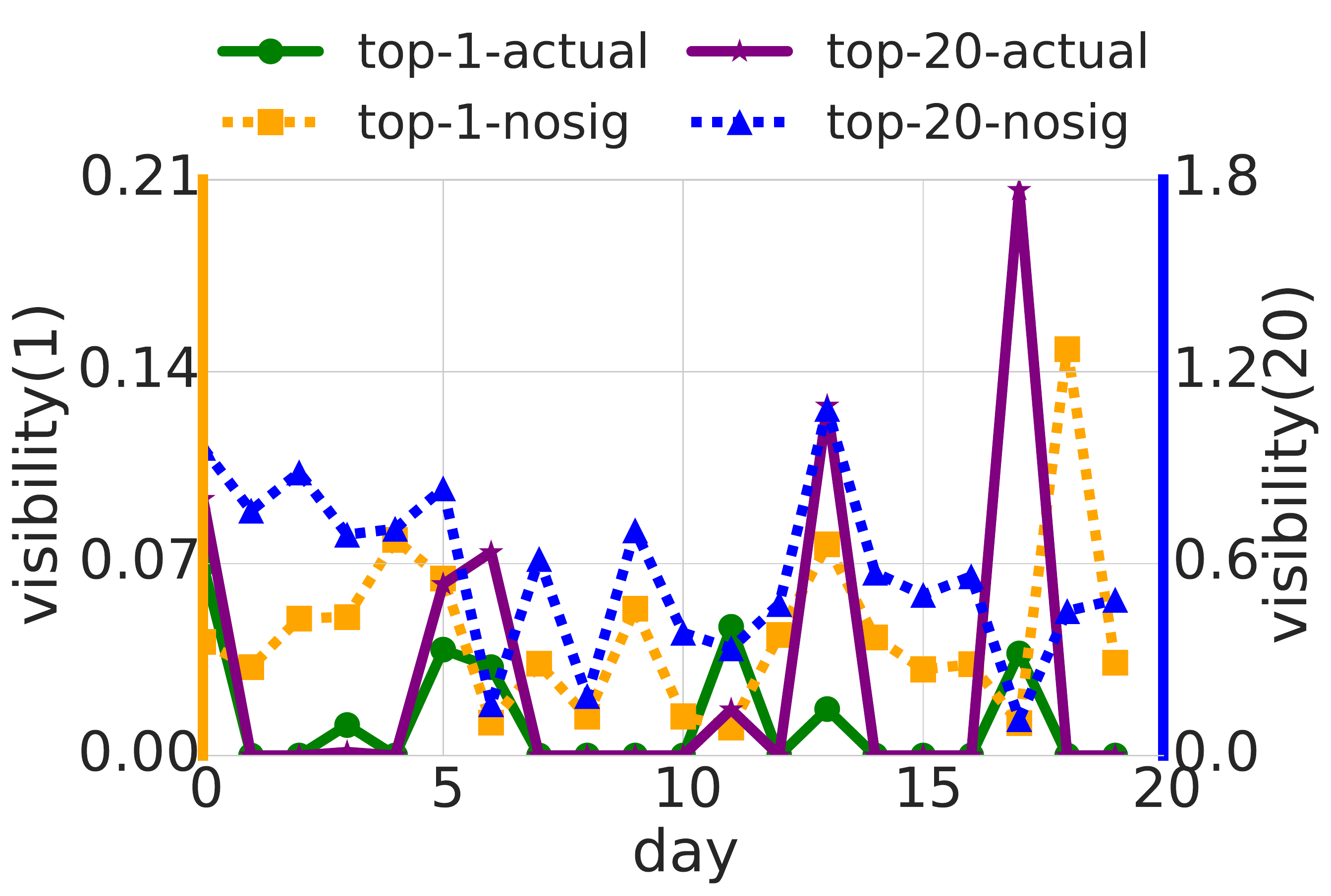} 
\end{tabular}
\vspace{-3mm}
\caption{Intensities and top-$k$ probabilities and visibilities. We focus on four broadcasters (one per column) and solve the AVM problem for one of their followers, picked at random. 
The first row shows the follower'{}s timeline intensity ($\mu(t)$, in brown) fitted using events from the training set, and the optimized intensities, as given by our framework, 
that maximize visibility for $k$$=$$1$, $20$ on the training set with and without significance ($\lambda^{*}(t)$, in solid and dashed yellow and blue, respectively).
The second row shows the top-$k$ probability for the optimized intensities with and without significance for $k$$=$$1$, $20$ ($f^{*}_k(t)$, in solid and dashed yellow and blue) 
as well as the follower'{}s significance ($s(t)$, in brown).
The third row compares the average visibility achieved by the optimized intensities without significance for $k$$=$$1$, $20$ ($\Vcal^{*}(k)$, in yellow and blue) to the average 
visibility achieved by the broadcaster'{}s posting activity ($\Vcal(k)$, in green and purple) on a held-out set.
}
\label{fig:explor-avmp}
\vspace{-2mm}
\end{figure*}

\xhdr{Dataset description and experimental setup}
We use data gathered from Twitter as reported in previous work~\cite{cha2010measuring}, which comprises the following 
three types of information: profiles of $52$ million users, $1.9$ billion directed follow links among these users, and $1.7$ billion public tweets posted 
by the collected users. The follow link information is based on a snapshot taken at the time of data collection, in September 2009.
Here, we focus on the tweets published during a six and a half month period, from February 2, 2009 to August 13, 2009. 
In particular, we sample $10$,$000$ users uniformly at random as broadcasters and record all the tweets they posted. Moreover, 
for each of these broadcasters, we track down all their followers and record all the tweets they posted as well as reconstruct 
their \emph{true} timelines by collecting all the tweets published by the people they follow.

In our experiments, we use the first three and a half month period, from February 2 to May 13 to fit the piecewise constant intensities of the 
followers'{} timelines and the followers'{} significance, which we use in our convex visibility shaping framework. 
Here, the follower'{}s significance is the probability that she is on-line, estimated as a piecewise (hourly) constant probability from the tweets-retweets the follower 
posted -- if a follower tweeted or retweeted in an hour, we assume it was on-line during that hour.
Then, we use the last three month period, from May 14 to August 13, to evaluate our framework.
We refer to the former period as the training set and the latter as the test set.
We ex\-pe\-ri\-ment both with $T$$=$$24$ hours ($M$$=$$24$, $\Delta$$=$$1$ hour) and $T$$=$$7$ days ($M$$=$$24 \times 7$, $\Delta$$=$$1$ 
hour), and set the budget $C$ to be equal to the average number of tweets per $T$ the broadcaster posted in the training period.
%
%
%

\xhdr{Evaluation schemes}
Throughout this section, we use three different evaluation schemes, with an increasing resemblance to a real world scenario:

\emph{Theoretical objective}:  We compute the theoretical value of the utility using the broadcaster intensity under study, be it the (optimal) intensity given by our convex visibility shaping
framework, the intensity given by an alternative baseline, or the the broadcaster'{}s (true) fitted intensity.

\emph{Simulated objective}:  We simulate events both from the broadcaster intensity under study and each of the followers'{} timeline fitted intensities. Then, we estimate empirically the overall
utility based on the simulated events. We perform $100$ independent simulation runs and report the average and standard error (or standard deviation) of the utility.

\emph{Held-out data}: We simulate events from the broadcaster intensity under study, interleave these generated events on the true followers'{} timelines recorded as test set, and compute
the corresponding utility. We perform $10$ independent simulation runs and report the average and standard error (or standard deviation) of the utility.

\xhdr{Intensities, top-$k$ probabilities and visibilities}
We pay attention to four broadcasters, picked at random, and solve the average visibility maximization task for one of their followers, also picked at random.
Our goal here is to shed light on the influence that the follower'{}s timeline intensity and significance have on the optimized broadcaster'{}s intensity as well
as its corresponding visibility and top-$k$ probability for different values of $k$.
Figure \ref{fig:explor-avmp} summarizes the results, which show that (i) including the significance in the visibility definition shifts the optimized intensities away 
from the times in which the followers are not online (first row); (ii) the optimized intensities typically achieve a higher average visibility than the one achieved 
by the broadcaster'{}s true posting activity on a held-out set (third row); and (iii) the optimized intensities are more concentrated in time for $k=1$ (first row) and
achieve a higher average visibility and top-$k$ probability for $k=20$ (second and third row).

\xhdr{Solution quality}
In this section, we perform a large scale evaluation of our framework across all $10$,$000$ broadcasters in terms of the three evaluation schemes 
described above and compare its performance against several baselines.
Here, we consider the definition of visibility that incorporates significance since, as argued previously, may lead to more effective broadcasting
strategies\footnote{We obtain qualitatively similar results if we omit the significance in the definition of visibility. Actually, in such case, our framework 
beats the baselines by a greater margin.}.

In the average visibility maximization task, we compare our framework with three heuristics, in which the broadcaster distributes 
the available budget uniformly at random (RAVM), proportionally to $\sum_{i=1}^n \mu_i(t)$ (IAVM) and proportionally to 
$\sum_{i=1}^n s_i(t) \mu_i(t)$ (PAVM), respectively.
In the minimax visibility maximization task, we also compare with three heuristics. The first two heuristics are similar to two of the ones 
just mentioned for AVM, \ie, the broadcaster distributes the available budget uniformly at random (RMVM) and proportionally to 
$\sum_{i=1}^n \mu_i(t)$ (IMVM).
In the third heuristic, the broadcaster distributes its budget following a greedy procedure: at each iteration $k$, it first finds the user with 
the least visibility given $\lambda^{(k-1)}(t)$ and then solves the average visibility maximization for that user given a budget of $C/n$. 
Finally, it outputs the intensity $\lambda(t) = \sum_{k=1}^{n} \lambda^{(k)}(t)$. The greedy procedure starts with $\lambda^{(0)}(t) = C/M$. 
Additionally, for the held-out comparison, we also compute the actual average intensity that the broadcaster achieved in reality.
\begin{figure}[t]
\centering
\begin{tabular}{c c c}
		   \hspace{-2mm}
		  { \footnotesize \rotatebox{90}{~~~~~Theoretical}} &
       		 \hspace{-2mm}\includegraphics[width=0.3\textwidth]{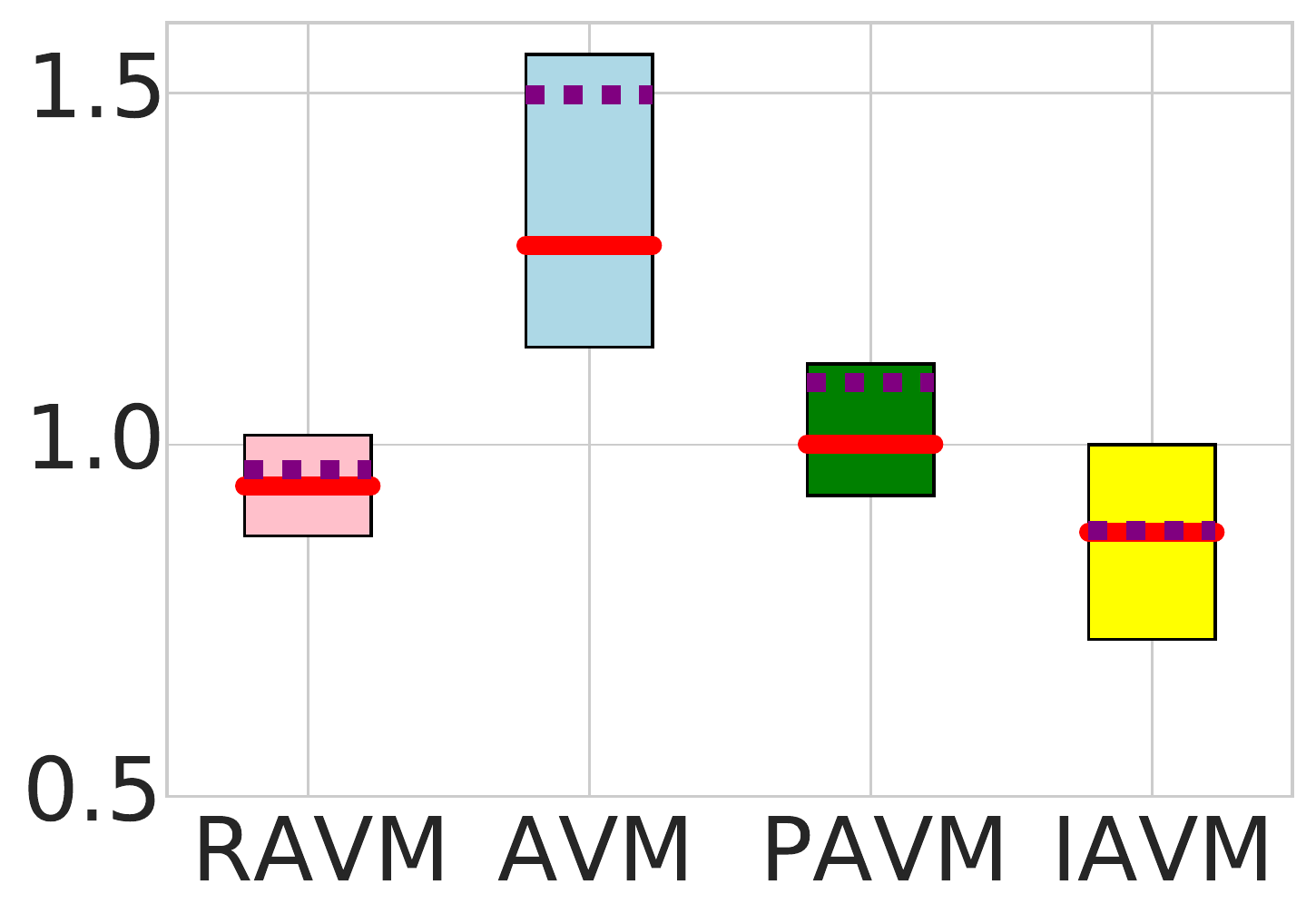} &
		 \includegraphics[width=0.3\textwidth]{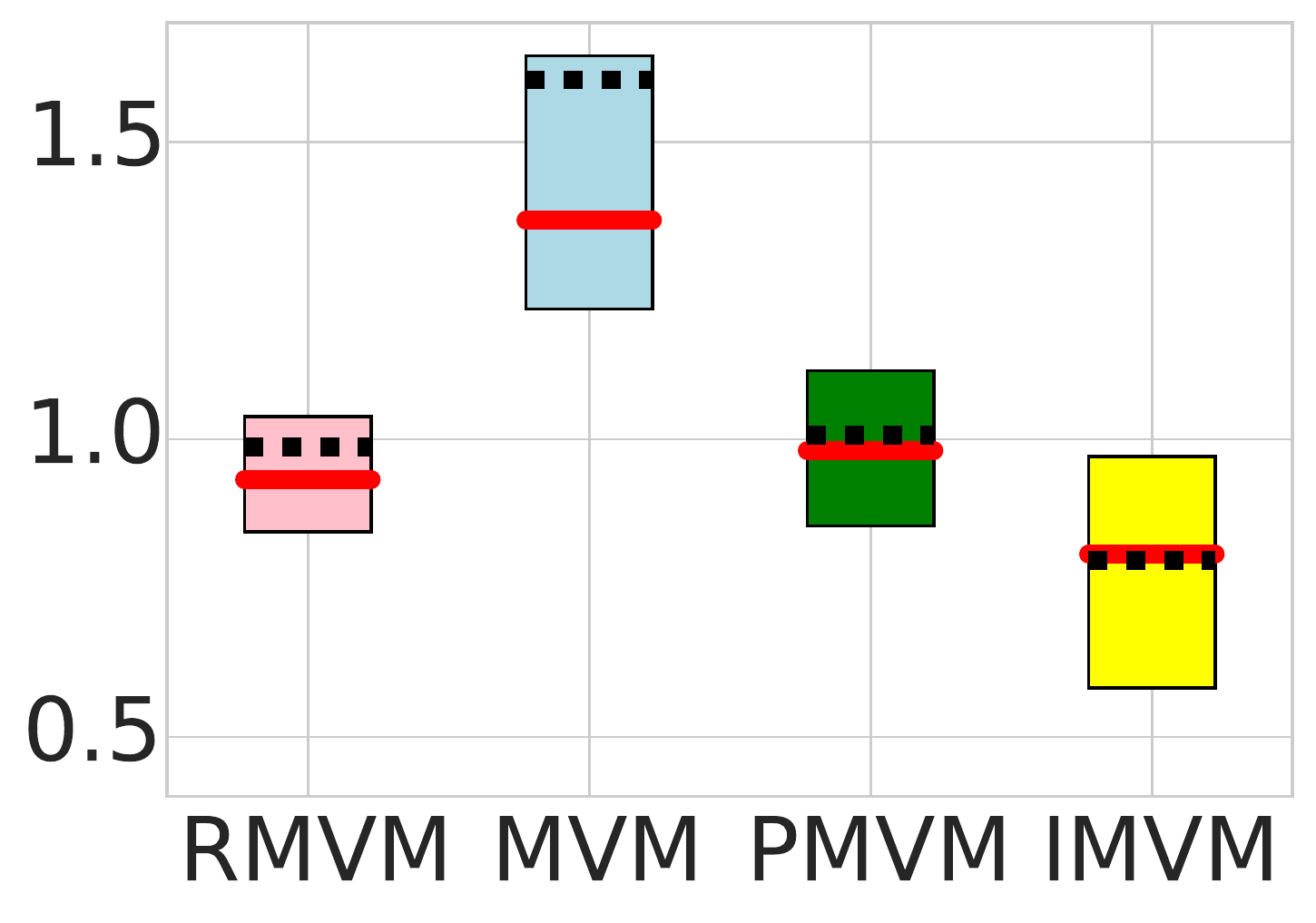} \\
		  \hspace{-2mm}
		 { \footnotesize \rotatebox{90}{~~~~~Simulated}} &
		 \hspace{-2mm}\includegraphics[width=0.3\textwidth]{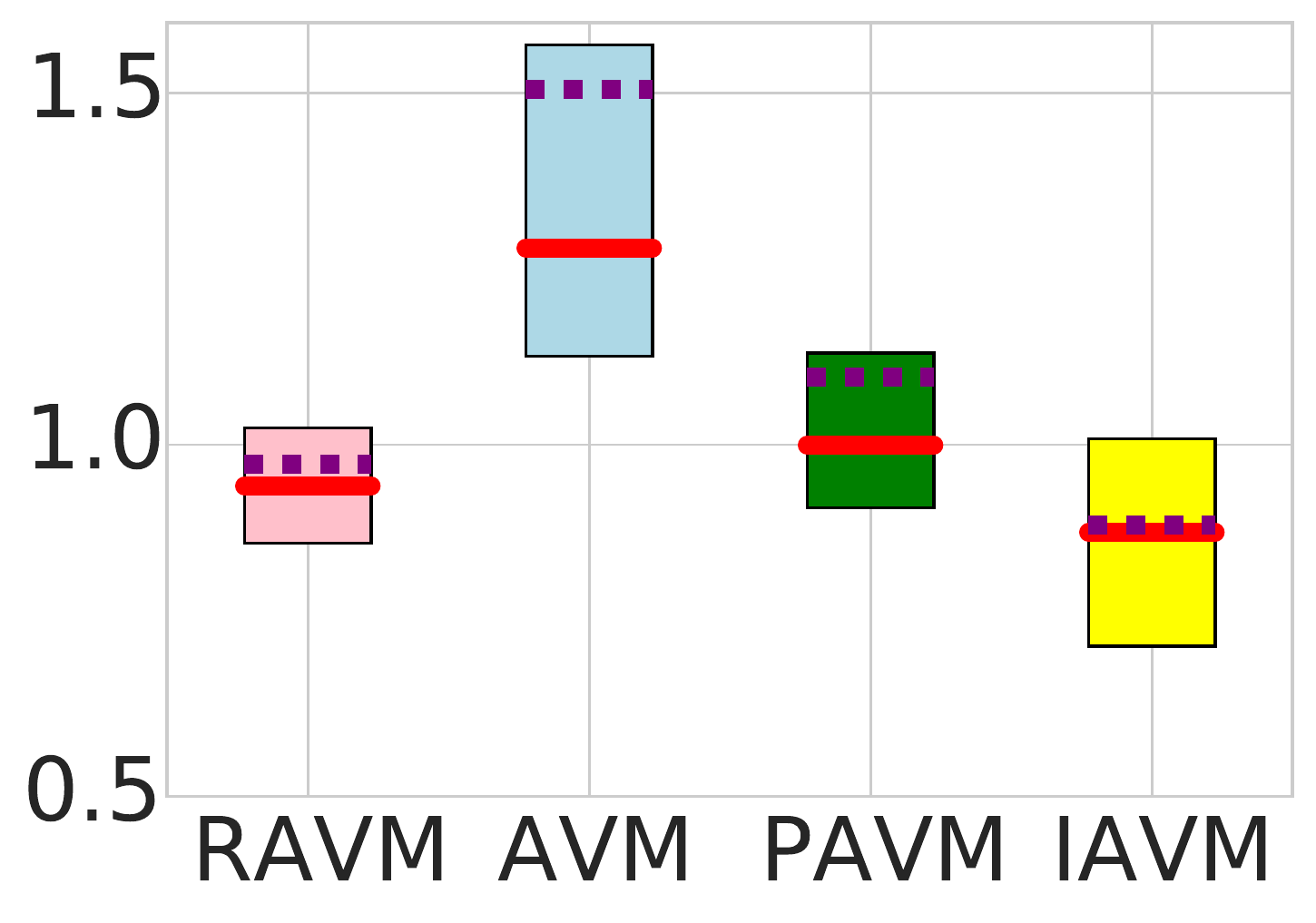} &
		  \includegraphics[width=0.3\textwidth]{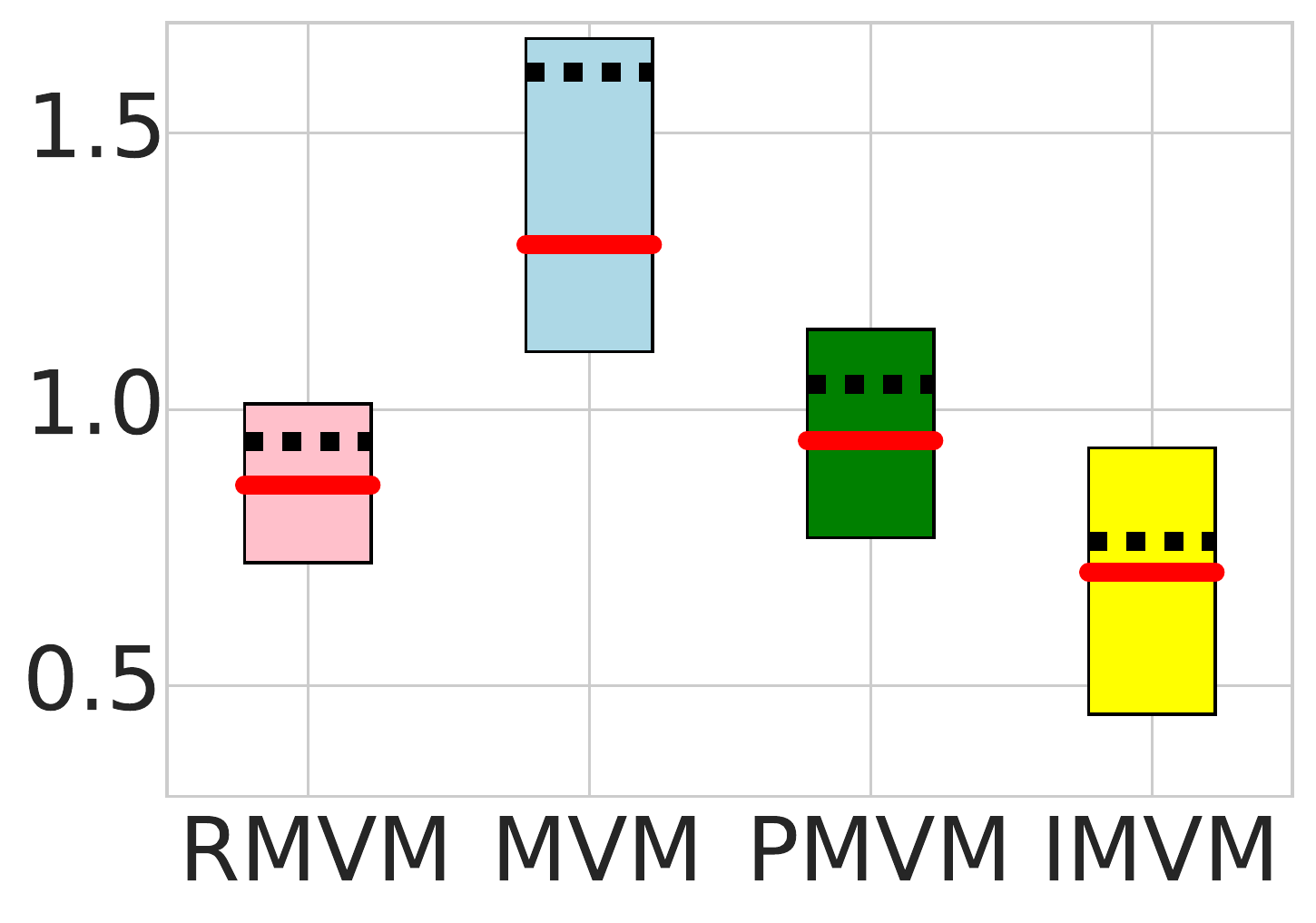} \\
		 \hspace{-2mm}
		  { \footnotesize \rotatebox{90}{~~~~Real Held-out}} &
		 \hspace{-2mm}\includegraphics[width=0.3\textwidth]{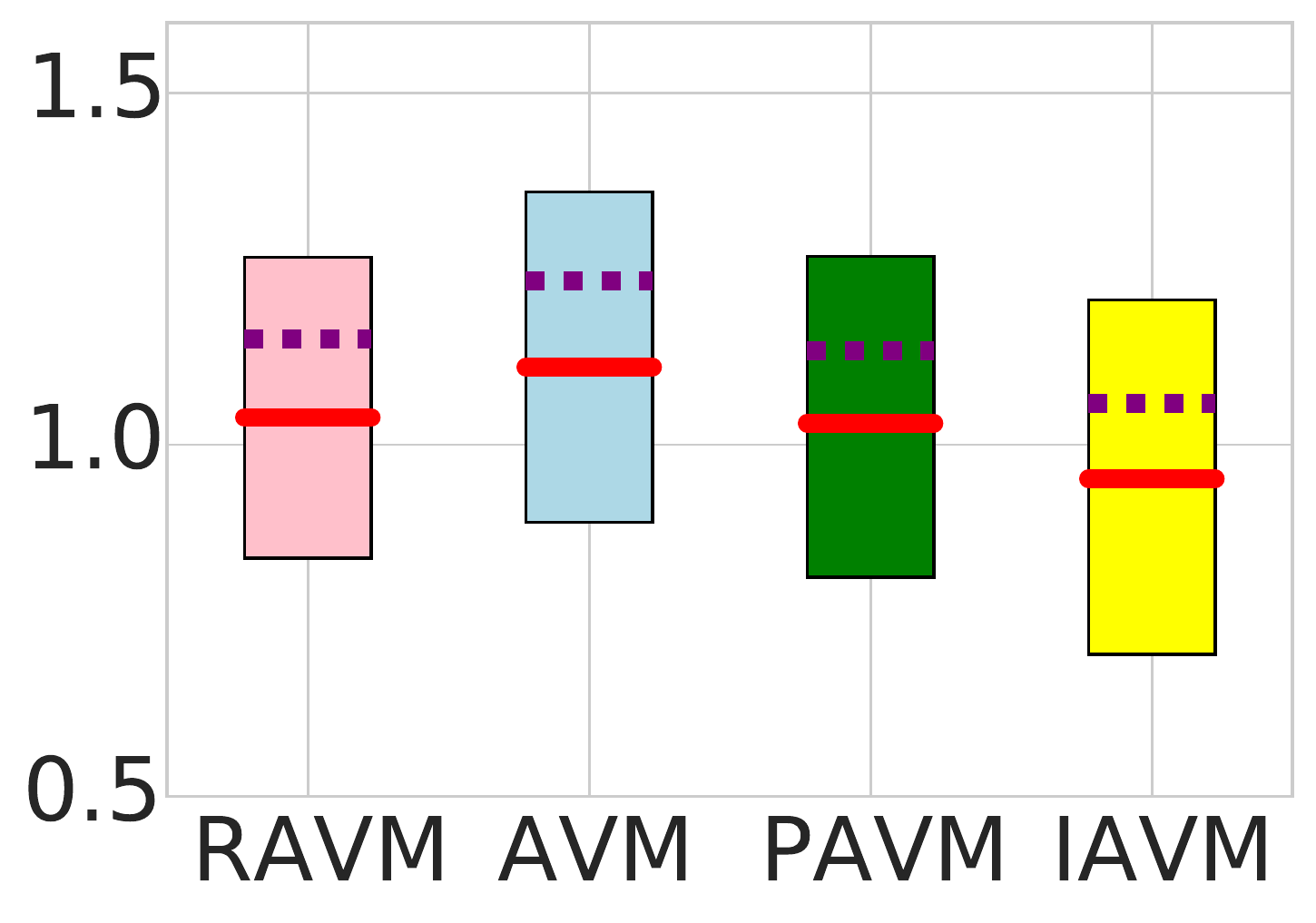} &
		 \includegraphics[width=0.3\textwidth]{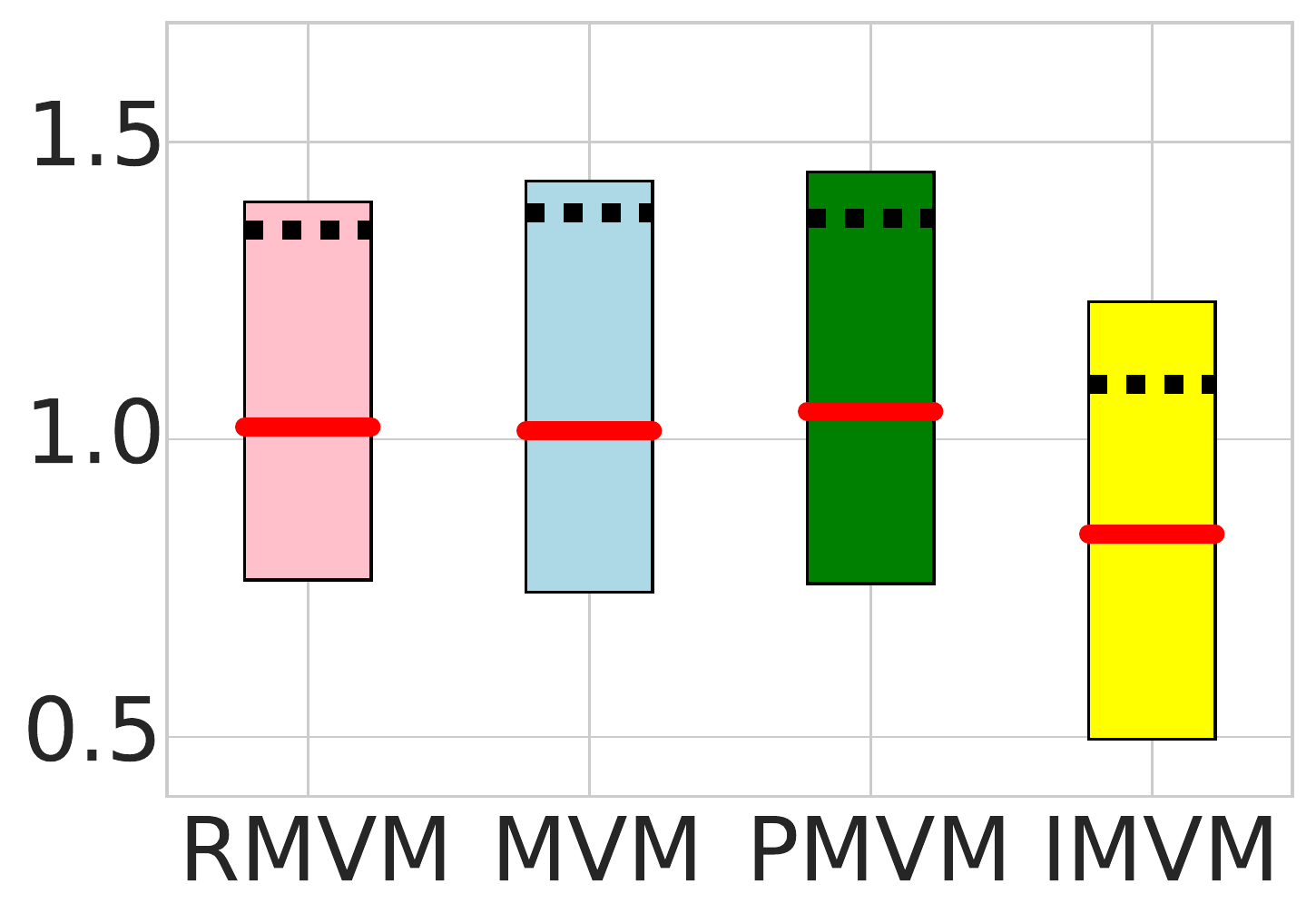} \\
		   & Average Visibility & Minimax Visibility
\end{tabular}
\caption{Visibility shaping for 10{,}000 broadcasters. The left (right) column corresponds to AVM (MVM), evaluated using the theoretical objective (first row),
the simulated objective (second row) and the held-out data (third row). The red (blue) dashed line shows the median (mean) objective popularity and the box 
limits correspond to the 25\%--75\% percentiles.}
\label{fig:compare-all}
\end{figure}

Figure~\ref{fig:compare-all} summarizes the results by means of a box plot, which shows the utilities achieved by our framework 
and the heuristics normalized with respect to the utility achieved by the broadcasters'{} fitted true intensity (by the posts during 
the test set for the third evaluation scheme).
That means, if $y=1$, the optimized intensity achieves the same utility as the broadcaster's{} recorded posts.
For the average visibility maximization task, the intensities provided by our method achieve $1.5$$\times$ higher theoretical objective and $1.3$$\times$ 
higher utility on a held-out set, in average (black dashed line), than the broadcaster'{}s fitted intensities. In contrast, alternatives fail at providing any 
gain, \ie, $y \leq 1$ for a half of the broadcasters.
Finally, for the minimax visibility maximization task, which is significantly harder, the intensities provided by our method achieve $1.6$$\times$ higher 
theoretical objective and $1.4$$\times$ higher average utility on a held-out set, in average (black dashed line), than the broadcaster'{}s fitted intensities.
In this case, although our method outperforms the baselines by large margins in terms of theoretical and simulated objectives, the baselines achieved
almost the same average utility on the held-out set.
The theoretical and simulated objective are almost equal in all cases, as one may have expected.
%
%

%

\xhdr{Solution quality vs. \# of followers}
Figure~\ref{fig:visibility-followers-k}(a) shows the average visibilities achieved by our optimized intensities for the AVM task, normalized by the average visibility that the corresponding broadcasters'{} 
fitted intensities achieve, against number of followers for the same $10$,$000$ broadcasters as above.
Independently of the number of followers, we find that the intensities provided by our method consistently outperform the broadcaster'{}s fitted intensities.

%
\vspace{-2mm}
\xhdr{Visibility vs. $k$}
Figure~\ref{fig:visibility-followers-k}(b) shows the average visibility achieved by our optimized intensities for the AVM task against $k$ for the four broadcasters 
from Figure~\ref{fig:explor-avmp}. 

\xhdr{Scalability}
Figure~\ref{fig:visibility-followers-k}(c) shows that our convex optimization framework easily scale to broadcasters with thousands of followers.
For example, given a broadcaster with $2{,}000$ followers, our algorithm takes $\sim$$250$ milliseconds to find the optimal intensity for the average 
visibility maximization using a single machine with $64$ cores and $1.5$ TB RAM.
\begin{figure}[t]
       \centering
       \begin{tabular}{c c c }
       \includegraphics[width=0.3\textwidth]{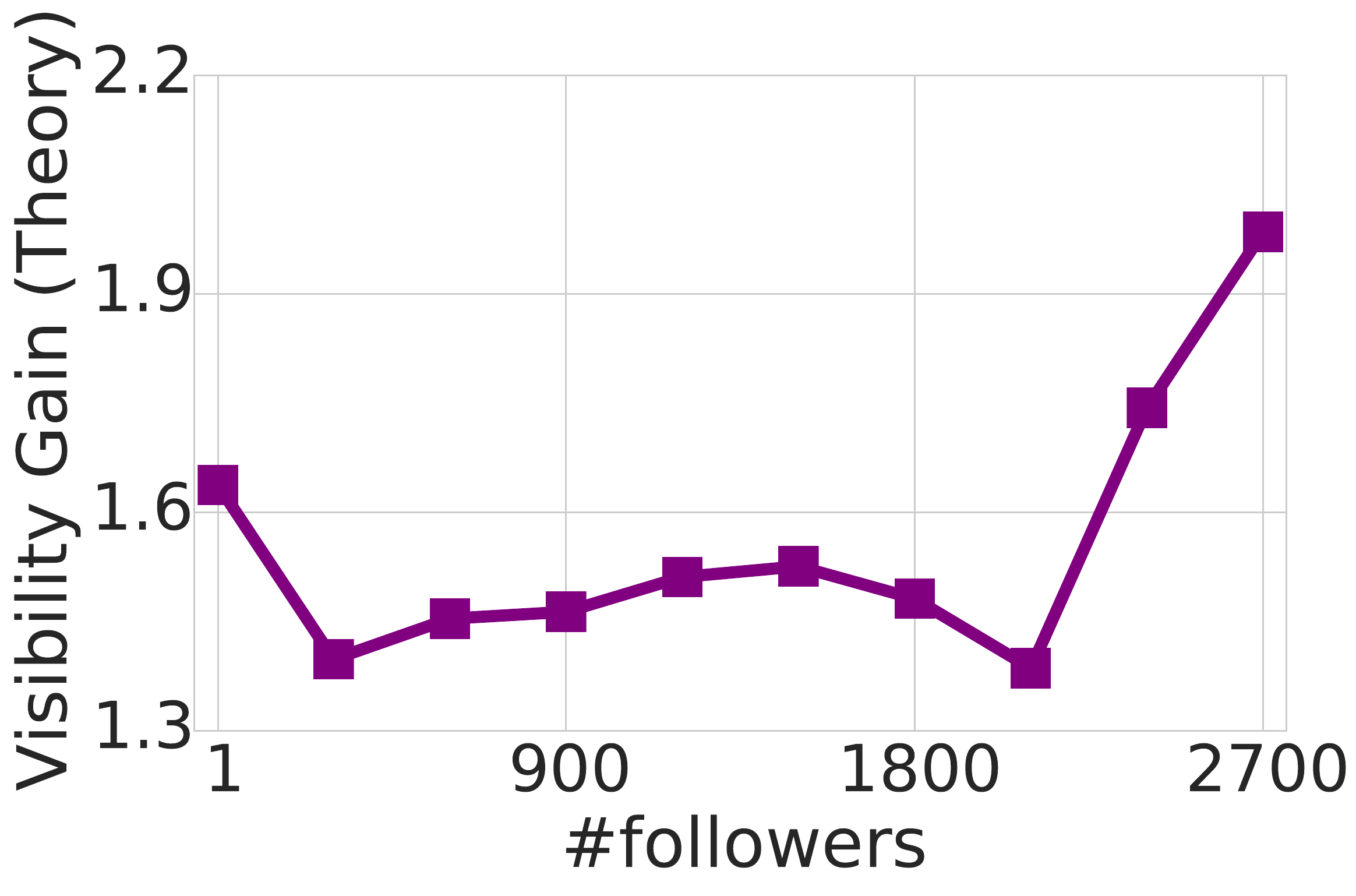} &
       \includegraphics[width=0.3\textwidth]{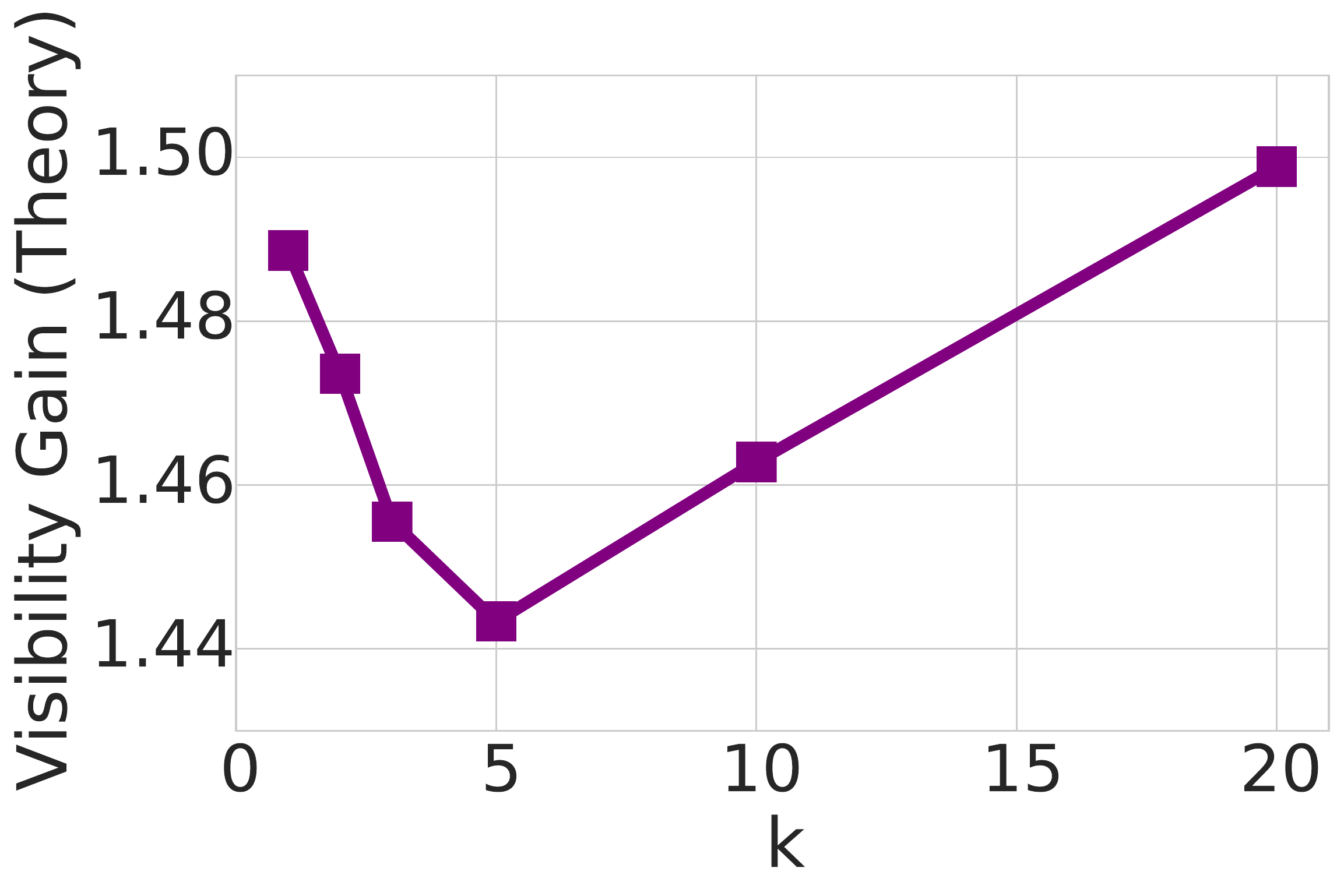} & 
       \includegraphics[width=0.3\textwidth]{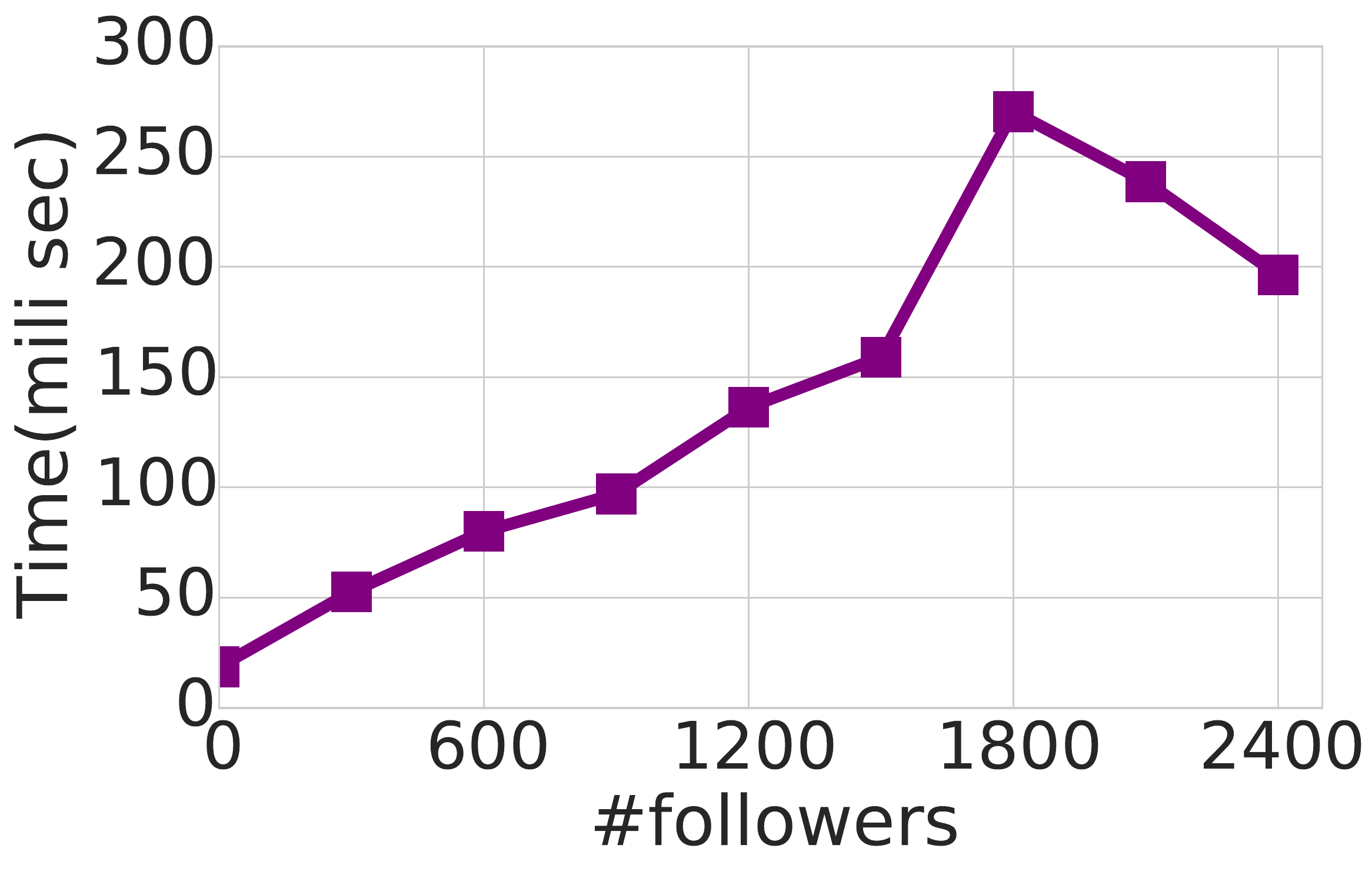}
       \\
       a) Followers & b) $k$ & c) Running time\\
       \end{tabular}
       \caption{Panels show the average visibility against (a) \# of followers and (b) $k$. Panel (c) plots running time.}
       \label{fig:visibility-followers-k}
\end{figure}

%% file: 080conclusions.tex
In this paper, we developed a novel framework to solve the when-to-post problem, in which we model users'{} feeds and posts
as discrete events occurring in continuous time. Under such continuous-time model, then choosing a strategy for a broadcaster
becomes a problem of designing the conditional intensity of her posting events.
The key technical idea that enables our framework is a novel formula which can link the conditional intensity of an \emph{arbitrary}
broadcaster with her visibility in her followers'{} feeds, defined as the time that at least one post from her is among the most recent $k$ 
received stories in her followers'{} feed.
In addition to the framework, we develop an efficient gradient based optimization algorithm, which allows us to find optimal broadcast
intensities for a variety of visibility shaping tasks in a matter of seconds.
Experiments on large real-world data gathered from Twitter revealed that our framework can consistently make broadcasters' posts more
visible than alternatives.

Our work also opens many interesting venus for future work. For example, we assume that the social network sorts stories in each user'{}s
feed in inverse chronological order. While this is a realistic assumption for some social networks (\eg, Twitter), there are other social networks
(\eg, Facebook) where the feed is curated algorithmically.
It would be very interesting to augment our framework to such cases.
In this work, we model users'{} intensities using inhomogeneous Poisson processes, whose intensities are history independent and deterministic.
Extending our framework to point processes with stochastic and history dependent intensity functions, such as Hawkes processes, would most
likely provide more effective broadcasting strategies.
In this work, we validate our framework on two visibility shaping tasks, average visibility maximization and minimax visibility maximization,
however, there are many other useful tasks one may think of, such as visibility homogenization.
Finally, it would be very interesting to investigate the scenario in which there are several smart broadcasters using our algorithm.

%% file: 090appendix.tex
\section{Proof of Lemma 1}
\label{sup:proof-lem-f_k}


We will prove this lemma by induction on $k$. For the case $k=1$, $f_1(t)$ satisfies a first-order linear differential equation,
\begin{align}
f_1'(t) =  -(\mu(t)+\lambda(t)) f_1(t) +  \lambda(t),
\end{align}
whose unique solution is
\begin{align} \label{eq:sol-eq-k1}
{f_1}(t) = \int_{0}^{t} \lambda(\tau) e^{-\int_{\tau}^{t} (\lambda + \mu)(x) dx}.
\end{align}
as long as assuming $f_1(t) = 0$. Then, using that $\Gamma [1,\int_{\tau}^{t} \mu(x) dx ] = e^{-\int_{\tau}^{t} \mu(x) dx}$, we can rewrite 
the solution as 
\begin{align} \nonumber
f_1(t) =  \int_{0}^{t} \lambda(\tau)  e^{-\int_{\tau}^{t} \lambda(x) dx} \Gamma  [1,\int_{\tau}^{t} \mu(x) dx] d\tau,
\end{align}
which proves the theorem for $k=1$. Now, in the inductive step we assume the hypothesis is true  for $1,2, \ldots k-1$ 
and we prove it for $k$.
We start by rewriting the differential given by Equation~\ref{eq:prob-diff-eq-top-k} as
\begin{align}
\label{eq:prob-diff-eq-top-k-appen}
{f_k}'(t) + (\mu(t)+\lambda(t)) f_k(t) =  \lambda(t) + \mu(t) f_{k-1}(t),
\end{align}
where, by assumption, $f_{k-1}(t)$ is unique and known. Then, as long as $f_k(0) = 0$, the above differential equation has a 
unique solution and thus we only need to find $f_k(t)$ that satisfies it. 
To do so, we rewrite the right hand side of the differential equation using the inductive hypothesis as
\begin{equation} \nonumber
\lambda(t) + \mu(t) \frac{ \int_{0}^{t}\lambda(\tau)  e^{-\int_{\tau}^{t} \lambda(x) dx}\Gamma[k-1,\int_{\tau}^{t} \mu(x) dx] d\tau} {(k-2)!}, \nonumber
\end{equation} 
which, using $\Gamma[k-1,x] = \frac{1}{k-1}(\Gamma[k,x] + \frac{\partial \Gamma[k,x]}{\partial x})$, can be expressed as
\begin{equation} \label{eq:diff-eq-gamma-property}
\lambda(t) + \mu(t)  \frac{\int_{0}^{t} \left(\lambda(\tau)  e^{-\int_{\tau}^{t} \lambda(x) dx}\Gamma[k,\int_{\tau}^{t} \mu(x) dx] + \frac{\partial \Gamma[k,\int_{\tau}^{t} \mu(x) dx}{\partial \int_{\tau}^{t} \mu(x) dx}\right) d\tau} {(k-1)!}
\end{equation}
Next, we hypothesize that
\begin{equation} \label{eq:hypothesized-fk}
f_k(t) = \frac{\int_{0}^{t} \lambda(\tau)  e^{-\int_{\tau}^{t} \lambda(x) dx}\, \Gamma[k,\int_{\tau}^{t} \mu(x) dx] \, d\tau)}{(k-1)!},
\end{equation}
%
%
and rewrite Eq.~\ref{eq:diff-eq-gamma-property} as
\begin{align}
\lambda(t) +\ \mu(t) f_k(t) + \frac{\mu(t) \int_{0}^{t}\lambda(\tau)  e^{-\int_{\tau}^{t} \lambda(x) dx} \frac{\partial \Gamma[k,\int_{\tau}^{t} \mu(x) dx]}{\partial \int_{\tau}^{t} \mu(x) dx} d\tau }{(k-1)!}. 
\end{align}
Then, by the fundamental theorem of calculus,
\begin{equation}
\frac{\partial \Gamma[k,\int_{\tau}^{t} \mu(x) dx]}{\partial \int_{\tau}^{t} \mu(x) dx} = 
\frac{ \partial \Gamma[k,\int_{\tau}^{t} \mu(x) dx]}{\partial t} \times \frac{\partial t}{\partial \int_{\tau}^{t} \mu(x) dx}
= \frac{\partial \Gamma[k,\int_{\tau}^{t} \mu(x) dx]}{\partial t} \times \frac{1}{\mu(t)} \nonumber
\end{equation}
and thus
\begin{equation}
\label{eq:induction-ode-simple-1}
 \lambda(t) + \mu(t) f_k(t) +\frac{\int_{0}^{t}\lambda(\tau)  e^{-\int_{\tau}^{t} \lambda(x) dx}\frac{\partial \Gamma[k,\int_{\tau}^{t} \mu(x) dx]}{\partial t} d\tau }{(k-1)!}.
\end{equation}
Finally, using that for differentiable functions $g$ and $h$, $gh' = (gh)' - g'h$, we have that
\begin{equation} \nonumber
\begin{split}
\int_{0}^{t}(\lambda(\tau)  e^{-\int_{\tau}^{t} \lambda(x) dx})(\frac{\partial \Gamma[k,\int_{\tau}^{t} \mu(x) dx]}{\partial t}) d\tau 
& = \underbrace{\int_{0}^{t} \frac{\partial (\lambda(\tau)  e^{-\int_{\tau}^{t} \lambda(x) dx}\Gamma[k,\int_{\tau}^{t} \mu(x) dx])}{\partial t}  d\tau}_{(k-1)! (f_k'(t) - \lambda(t))}
\\ & - \underbrace{\int_{0}^{t}  (\frac{\partial \lambda(\tau)  e^{-\int_{\tau}^{t} \lambda(x) dx}}{\partial t})(\Gamma[k,\int_{\tau}^{t} \mu(x) dx]) d\tau}_{(k-1)! \lambda(t) f_k(t)}
\end{split}
\end{equation}
%
%
%
and then we can rewrite Eq.~\ref{eq:induction-ode-simple-1} as
\begin{align}
\lambda(t) + \mu(t) f_k(t) + \frac{(k-1)! (f_k'(t) - \lambda(t)) +  (k-1)! \lambda(t) f_k(t) }{(k-1)!}, \nonumber
\end{align}
which simplifies to
\begin{align}
{f_k}'(t) + (\mu(t)+\lambda(t)) f_k(t). \nonumber \nonumber
\end{align}
This asserts that hypothesized solution for $f_k(t)$ in Eq.~\ref{eq:hypothesized-fk} satisfies Eq.~\ref{eq:prob-diff-eq-top-k-appen}, hence, 
it is the unique solution for $f_k(t)$.

\section{Proof of Theorem 3} \label{sup:concave-f-k}
From Lemma~1, we know that
\begin{align}
 f_k(t) =  \frac{\int_{0}^{t}(\lambda(\tau)  e^{-\int_{\tau}^{t} \lambda(x) dx})\Gamma[k,\int_{\tau}^{t} \mu(x) dx] d\tau} {(k-1)!}. \nonumber
\end{align}
Using integration by parts, we can rewrite the above ex\-pression as
\begin{equation}
 f_k(t) = 1 - \frac{e^{-\int_{0}^{t} \lambda(x) dx} \Gamma[k,\int_{0}^{t} \mu(x) dx]}{(k-1)!} 
- \frac{\int_{0}^{t} (e^{-\int_{\tau}^{t} \lambda(x) dx})\frac{\partial \Gamma[k,\int_{\tau}^{t} \mu(x) dx]}{\partial \tau} d\tau}{(k-1)!} \nonumber
\label{eq:expand $f_k(t)$}
\end{equation}
Lemma~\ref{lem:convexity-exp} tells us that $e^{-\int_{0}^{t} \lambda(x) dx}$ and $e^{-\int_{\tau}^{t} \lambda(x) dx}$ are convex with
respect to $\lambda(\cdot)$. 
Moreover, using Lemma~\ref{lem:convexity-integral} and the fact that $\frac{\partial \Gamma[k,\int_{\tau}^{t} \mu(x) dx]}{\partial \tau} > 0$,
it follows that the function ${\int_{0}^{t} (e^{-\int_{\tau}^{t} \lambda(x) dx})\frac{\partial \Gamma[k,\int_{\tau}^{t} \mu(x) dx]}{\partial \tau} d\tau}$ is 
convex. 
Finally, given that $\Gamma[k,\int_{0}^{t}  0] > 0$, we can conclude that $f_k(t)$ is concave with respect to $\lambda(\cdot)$. 

\begin{lemma}
\label{lem:convexity-exp}
Functional $J[\lambda] = e^{- \int_a^t \lambda(x) dx}$ is convex with respect to $\lambda(\cdot)$ for any constant $a \leq t$.
\end{lemma}
\begin{proof}
We simply verify that $J[\lambda]$ satisfies the definition of convexity, as given by Eq.~\ref{eq:convexity-def}:
\begin{equation}
J[\alpha \lambda_1 + (1-\alpha) \lambda_2 ] = e^{- \int_a^t \alpha \lambda_1(x) + (1-\alpha) \lambda_2(x) \, dx}
\le \alpha e^{- \int_a^t  \lambda_1(x)\, dx}+  (1-\alpha) e^{- \int_a^t  \lambda_2(x) \, dx} 
= \alpha J[\lambda_1] + (1-\alpha) J[\lambda_2] \nonumber
\end{equation}
where the inequality follows from the arithmetic-geometric mean inequality, \ie, $\theta x + (1-\theta) y \ge x^{\theta} y^{1-\theta}$ for all positive $x$, $y$, and $0 < \theta < 1$.
\end{proof}

\begin{lemma}
\label{lem:convexity-integral}
If the functional $J_t[\lambda(\cdot)] $ is convex with respect to $\lambda(\cdot)$. Then, given any arbitrary function $g(x) \ge 0$, the functional 
$L[\lambda] = \int_0^t J_{\tau}[\lambda(.)] g(\tau) d\tau $ is also convex with respect to $\lambda(\cdot)$.
\end{lemma}
\begin{proof}
We verify that the functional $L[\lambda] = \int_0^t J_{\tau}[\lambda(.)] g(\tau) d\tau $ verifies the definition of convexity, as given
by Eq.~\ref{eq:convexity-def}:
\begin{align}
L[\alpha \lambda_1 + (1-\alpha) \lambda_2 ] &= \int_0^t J_{\tau}[\alpha \lambda_1 + (1-\alpha) \lambda_2] g(\tau) d\tau
\le  \alpha \int_0^t J_{\tau}[\lambda_1] g(\tau) d\tau
+ (1-\alpha) \int_0^t J_{\tau}[\lambda_2] g(\tau) d\tau \nonumber \\ 
&= \alpha L[\lambda_1] + (1-\alpha) L[\lambda_2] \nonumber
\end{align}
where the inequality holds using that, given any two arbitrary functions $h_1$ and $h_2$ such that $h_1(x) \ge h_2(x) \ge 0$ for all $x \in \Dcal$, then 
$\int_{\Dcal} h_1(x) g(x) dx \ge \int_{\Dcal} h_2(x) g(x) dx$ given $g(x) \ge 0$ for all $x \in \Dcal$.
\end{proof}

\section{Proof of Theorem 4} \label{sup:concave-expected-T}

Theorem~\ref{theo:concave-f-k} proves the concavity of $f_k(t)$ with respect to $\lambda(\cdot)$. 
Therefore, $1-f_k(t)$ is convex and, using Lemma~\ref{lem:convexity-integral}, it holds that $\int_0^T (1-f_k(t)) s(t) dt = \int_0^T s(t) dt - \int_0^T f_k(t) s(t) dt$ is also convex. 
Then, since $\int_0^T s(t) dt$ is constant, $\int_0^T f_k(t) s(t) dt$ is concave with respect to $\lambda(\cdot)$ and the proof is complete.

\section{Proof of Lemma 6} \label{sup:lem-limit}
Each piecewise continues function can be represented as summation of a number of \emph{heaviside} step functions. The count is equal to the number of discontinuity points. However,  each heaviside function itself is the limit of smooth \emph{tanh} functions. Therefore, the piecewise continues function will be the limit of a finite summation of smooth \emph{tanh} functions.

\section{Proof of Theorem 7} \label{app:piecewise-concave}
Consider two piecewise constant  functions $\lambda(\cdot) ,\mu(\cdot) \in \Gcal$.
According to Lemma~\ref{lem:sequence} there exist sequence of smooth functions such that $\lim_{n \to \infty} \lambda_n = \lambda$  and $\lim_{n \to \infty} \lambda'_n = \lambda'$. Because of the concavity 
of $f_k$ in $\Hcal$ we know for $ 0 < \alpha < 1$:
\begin{align*}
f_k [\alpha \lambda_n(\cdot) + (1-\alpha) \lambda'_n(\cdot)]  \ge \alpha f_k[\lambda_n(\cdot)] + (1-\alpha) f_k[\lambda'_n(\cdot)].
\end{align*}
Taking the limit and using the continuity of $f_k$ we get:
\begin{equation}
f_k [\alpha \lambda(\cdot) + (1-\alpha) \lambda'(\cdot)]  \ge \alpha f_k[\lambda(\cdot)] + (1-\alpha) f_k[\lambda'(\cdot)].
\end{equation}
Accompanied with convexity of space $\Gcal$ the theorem is proved.